%% file: main.tex
\documentclass[12pt]{article}
\pdfoutput=1

\usepackage{amsmath,amssymb,amsfonts}
\usepackage{longtable}
\usepackage{tabularx}

\usepackage{float}

\usepackage{amsthm}

\usepackage{a4wide,epsfig,psfrag,scalefnt}
\usepackage[dvipsnames]{xcolor}
\usepackage{braket}
\usepackage{placeins}

\usepackage{tensor}

\usepackage{slashed}
\usepackage{enumitem}
\usepackage{cite}
\usepackage[numbers,sort&compress]{natbib}
\usepackage{caption}
\usepackage{subcaption}
\usepackage{verbatim}
\usepackage{xcolor}
\usepackage{physics}

\usepackage{pdfpages}
\usepackage{diagbox}
\usepackage{array}

\usepackage{hyperref}

\graphicspath{ {figs/} }

\input{commands}

\allowdisplaybreaks

\begin{document}

\title{\vskip-3cm{\baselineskip14pt
    \begin{flushleft}
     \normalsize CERN-TH-2026-xxx, DESY-26-013
    \end{flushleft}} \vskip2.5cm
  Heavy-quark production in deep-inelastic scattering \\
  -- Mellin moments of structure functions -- 
  }

\author{
  Marco Klann$^{a}$,
  Sven-Olaf Moch$^{a}$, 
  Kay Sch\"onwald$^{b}$
  \\[3ex]
  {\small\it (a) II. Institut f\"ur Theoretische Physik, Universit\"at Hamburg, D-22761 Hamburg, Germany}
  \\
  {\small\it (b) Theoretical Physics Department, CERN, 1211 Geneva 23, Switzerland}
}

\date{}

\maketitle

\begin{abstract}

\noindent
We compute Mellin moments of the heavy‑quark structure functions in deep-inelastic scattering at next‑to‑leading order in quantum chromodynamics, retaining their full dependence on the heavy‑quark mass. Using the optical theorem and the operator product expansion, we derive analytic results for fixed Mellin moments $N = 2$ to 22 of the structure functions $F_2$ and $F_L$. 
Our results reproduce the known expressions in the relevant asymptotic limits, in particular for 
virtualities of the exchanged photon $Q^2$ much larger than the heavy‑quark mass squared $m^2$, and are in agreement with existing parametrisations of the next‑to‑leading‑order coefficient functions.
The computational set‑up developed in this work also provides a direct pathway toward extending these calculations to next‑to‑next‑to‑leading order.

\vspace*{5em}

\end{abstract}

\thispagestyle{empty}

\thispagestyle{empty}

\newpage

\section{Introduction}
\label{sec:intro}

Heavy‑quark production in deep‑inelastic scattering (DIS) plays a central role in testing quantum chromodynamics (QCD) and in determining parton distribution functions (PDFs) at high energies. 
In neutral‑current (NC) DIS, heavy quarks such as charm or bottom are produced predominantly through photon–gluon fusion, in which a highly virtual photon with momentum $q$ interacts with a gluon inside the nucleon. 
The presence of the heavy‑quark mass $m$ introduces a hierarchy of scales, making these processes a sensitive probe of perturbative QCD. 
Measurements of charm and bottom production at the HERA collider (see Ref.~\cite{H1:2018flt}) have established heavy‑quark DIS as an essential component of precision QCD phenomenology and have provided stringent constraints on the gluon PDF down to small values of the Bjorken scaling variable $x$. 
Further high‑precision data are anticipated from the future Electron–Ion Collider (EIC) (see Ref.~\cite{AbdulKhalek:2021gbh}), which will explore a similar kinematic range while significantly deepening the experimental understanding of heavy‑flavour production.

The cross section for unpolarised heavy-quark production in DIS for NC exchange by photons of virtuality $Q^2 = -q^2 > 0$ is expressed in terms of the heavy-flavour structure functions $F_k(x,Q^2,m^2)$ with $k=2,L$. 
Standard collinear QCD factorisation expresses the structure functions as convolutions of the PDFs $f_{i}^{}$  with the corresponding coefficient functions $C_{k, i}^{}$.
The latter depend on the renormalisation scale $\mu_r$ and the factorisation scale $\mu_f$, which control the running of the strong coupling and the separation between long‑ and short‑distance physics, respectively. Thus,
\begin{equation}
  \label{eq:totalF2c}
  F_k^{}(x,Q^2,m^2) \;=\; 
  \sum\limits_{i \,\in\, \{ {\rm q,{\bar{q}},g} \}} \,\,
  \int_{x}^{\,z^{\,\rm max}}
  {dz \over z} \: f_{i}^{}\left({x \over z},\, \mu_f^2 \right)\,
  C_{k, i}^{}\left(z,\,Q^2,\,m^2,\,\mu_f^2,\,\mu_r^2 \right)
  \; ,
\end{equation}
where $z$ is the parton momentum fraction bounded by $z^{\,\rm max} = 1/(1 + 4\:\! m^2/Q^2)$.
The partonic kinematic variables (including the partonic centre-of-mass energy $s$)
are defined as follows:
\begin{align}
\label{eq:kinematics}
s = Q^2\left(\frac{1}{z}-1\right)\, , \quad
\xi = \frac{Q^2}{m^2}\, , \quad
\eta = \frac{s}{4m^2} - 1 \, ,
\end{align}
see for instance Ref.~\cite{Riemersma:1994hv}.
In this approach, the heavy quarks appear only as final‑state particles, and leading‑twist factorization ensures that the kinematic conditions for intrinsic charm are not satisfied (see e.g., Ref.~\cite{Blumlein:2015qcn}).

QCD perturbation theory provides a systematic expansion of observables in powers of the strong coupling $\alpha_s$, allowing the DIS coefficient functions in Eq.~(\ref{eq:totalF2c}) to be computed with high precision. 
In the heavy‑quark sector, this requires retaining the full dependence on the heavy‑quark mass, and the complete QCD corrections to these coefficient functions are known at next‑to‑leading order (NLO) in QCD in numerical form; see Refs.~\cite{Laenen:1992zk, Riemersma:1994hv, Harris:1995tu}.
These results have also been used to construct a parametrization in Mellin space; see Ref.~\cite{Alekhin:2003ev}.
Analytic results at NLO have first been presented in the limit $Q^2 \gg m^2$ (see Ref.~\cite{Buza:1995ie}) 
followed by Refs.~\cite{Blumlein:2016xcy,Blumlein:2019qze}, which contain the complete results with full mass dependence for the quark-flavour non-singlet and pure-singlet coefficient functions in Eq.~(\ref{eq:totalF2c}).
Analytic results with full mass dependence for the NLO gluon coefficient functions $C_{k,{\rm g}}$ 
are unknown.

Current global PDF fits (see, e.g. Ref.~\cite{Alekhin:2017kpj}) show that the precision of heavy‑flavour measurements from HERA, together with the anticipated accuracy of future EIC data, requires substantially improved theoretical predictions. 
The experimental uncertainties are now small enough that NLO accuracy no longer matches the available precision, especially in regions sensitive to the small‑$x$ gluon PDFs. 
To fully exploit these data, next‑to‑next‑to‑leading‑order (NNLO) QCD calculations are required. 
At NNLO accuracy there are now complete results in the limit $Q^2 \gg m^2$ available (see Refs.~\cite{Ablinger:2025joi,Ablinger:2025awb} building on Refs.~\cite{Ablinger:2010ty,Ablinger:2014lka,Behring:2014eya,Ablinger:2014vwa,Ablinger:2014nga,Ablinger:2019etw,Behring:2021asx,Blumlein:2021xlc,Ablinger:2022wbb,Ablinger:2023ahe,Ablinger:2024xtt,Behring:2025avs}), which can be combined with other partial information on the  coefficient functions in the threshold region $s \simeq 4\:\! m^2$,
cf. Refs.~\cite{Laenen:1998kp,Kawamura:2012cr} and in 
the high-energy regime $s \gg m^2$, cf. Ref.~\cite{Catani:1990eg}.

We aim to extend the heavy‑quark DIS coefficient functions to NNLO, and as a first step toward this goal we compute Mellin moments of the DIS structure functions while retaining full heavy‑quark mass dependence. 
To set up the formalism, we report here on the calculation of these moments at NLO and confirm agreement with existing results.
In particular, this allows us to verify the accuracy of the parametrisations of the NLO coefficient functions provided in Ref.~\cite{Riemersma:1994hv}.

The article is organised as follows. In Sec.~\ref{sec:set-up} we set the stage and lay the groundwork for the calculation of the heavy-quark DIS coefficient functions at NLO in QCD. 
In Sec.~\ref{sec:hq-mellin-moments} we compute their Mellin moments for fixed values of $N$ using the harmonic projection and an expansion of the forward Compton amplitude in the external momenta. 
This approach encounters bottlenecks for large $N$ in the integral reductions, so that we also adopt an alternative method, performing the integral reductions without relying on a momentum expansion. 
Sec.~\ref{sec:results} presents the results and discusses the various checks performed. 
We conclude in Sec.~\ref{sec:conclusions} with an outlook on the next steps toward NNLO accuracy. 
Appendix~\ref{sec:appA} provides details on the flavour factors relevant for the heavy‑quark structure functions.

\section{Set-up}
\label{sec:set-up}

\subsection{Hadronic Tensor}

This section and the subsequent ones lay the groundwork for the calculation of fixed Mellin moments of the coefficient functions.
It closely follows the presentation in Refs.~\cite{Moch:1999eb,Larin:1996wd}. 
For a more detailed review on the hadronic tensor, operator product expansion and renormalisation the reader is referred to Ref.~\cite{Buras:1979yt}.
To this end, the hadronic part of lepton-hadron scattering is considered, allowing for both light-quark fields $\psi_\mathrm{l}$ and heavy-quark fields $\psi_\mathrm{h}$.
These fields are coupled to external vector bosons, such as the photon, and the corresponding interactions are determined by a locally conserved current,
\begin{align}
	\tensor{J}{_\mu} = \bar{\psi}_\light Q^\vector_\light  \tensor{\gamma}{_\mu} \psi_\light + \bar{\psi}_\heavy Q^\vector_\heavy \tensor{\gamma}{_\mu} \psi_\heavy \, ,
\end{align}
where $Q^\vector_\light$ and $Q^\vector_\heavy$ denote the charge matrices for light- and heavy-quark flavours, respectively.

In the context of deep-inelastic scattering, the focus is typically on unpolarised hadronic states with momentum $p$ and spin $s$, representing particles such as the proton.
The spin-averaged hadronic matrix element of a product of vector currents is given by
\begin{align}
	\bra{p,\text{hadron} } \tensor{J}{^\dagger_\mu} (z) \tensor{J}{_\nu} (0) \ket{p,\text{hadron}} \equiv \frac{1}{\abs{\mathcal{S}}} \sum_{s\in \mathcal{S}} \bra{p,s,\text{hadron}} \tensor{J}{^\dagger_\mu} (z) \tensor{J}{_\nu} (0) \ket{p,s,\text{hadron} }\, .
\end{align}
Here, $\mathcal{S}$ denotes the set of all possible spin states of the hadron and 
$\abs{\mathcal{S}}$ is its cardinality.
The hadronic tensor is the Fourier transform of the spin-averaged matrix element
\begin{align}
   \tensor{W}{_{\mu\nu}}(p,q) = \frac{1}{4\pi} \int \mathrm{d}^4z \, e^{i q\cdot z} \bra{p}  \tensor{J}{^\dagger_\mu} (z) \tensor{J}{_\nu} (0) \ket{p}\, .
\end{align}
It should be noted that the hadronic tensor also depends on the heavy-quark mass $m$ and the strong coupling $\alpha_s$.
By convention, however, functional dependencies arising from parameters of the Lagrangian are suppressed unless they play an essential role in the discussion.

The hadronic tensor is constrained by Lorentz symmetry, parity, time-reversal invariance, and current conservation, and can therefore be written as
\begin{align}
    \tensor{W}{_\mu_\nu}(p,q) = \frac{1}{2x}
 \tensor{e}{_\mu_\nu} (p,q) F_L (p,q) + \frac{1}{2x} \tensor{d}{_\mu_\nu} (p,q) F_2 (p,q)\, ,
\end{align}
where the $F_i$ denote the structure functions, accompanied by symmetric and transverse tensors
\begin{subequations}
\begin{align}
     \tensor{e}{_\mu_\nu} (p,q) &=  \tensor{g}{_\mu_\nu} -\frac{\tensor{q}{_\mu} \tensor{q}{_\nu}}{\sp{q}{q}} \comma \LineNumber 
     \tensor{d}{_\mu_\nu} (p,q) &= - \frac{2\bjorken}{\sp{q}{q}} (\tensor{p}{_\mu} \tensor{q}{_\nu}+ \tensor{q}{_\mu} \tensor{p}{_\nu})- \frac{4\bjorken^2}{\sp{q}{q}}  \tensor{p}{_\mu} \tensor{p}{_\nu} -  \tensor{g}{_\mu_\nu} \period
\end{align}
\end{subequations}

\subsection{Forward Scattering Amplitude}
The exponential in the hadronic tensor becomes highly oscillatory in the deep-inelastic limit $-q^2 \to \infty$.
By the Riemann-Lebesgue lemma (see Refs.~\cite{rudin1987real,Muta:2010xua}), the hadronic tensor is then dominated by contributions on the light cone, $z^2 = 0$, which is also referred to as light-cone dominance (see Refs.~\cite{Jaffe:1972uq,Brandt:1970kg}).

The hadronic tensor can be related to the discontinuity of the forward scattering amplitude. 
Assuming that the momentum transfer $q$ is space-like and that the hadronic tensor is sufficiently well-behaved for the Schwarz reflection principle (or suitable generalisations) to hold, it can be expressed in terms of the imaginary part of the forward scattering amplitude.
For a modern account of the relation between imaginary parts and discontinuities, the reader is referred to Ref.~\cite{Hannesdottir:2022bmo}.
This relation is an instance of the optical theorem
\begin{align}
   \tensor{W}{_{\mu\nu}}(p,q) = \frac{1}{2\pi} \Im \, \tensor{T}{_{\mu\nu}}(p,q) \comma
\end{align}
where the forward scattering amplitude is defined (using time-ordering $\mathcal{T}$) by
\begin{align}
    \tensor{T}{_{\mu\nu}}(p,q) =
    i \int \mathrm{d}^4z \, e^{i q\cdot z} \bra{p}  \mathcal{T} \{ \tensor{J}{^{\dagger}_\mu}(z) \tensor{J}{_\nu}(0) \} \ket{p} \period 
\end{align}

Through the optical theorem, the hadronic tensor is directly connected to the operator product expansion of the forward scattering amplitude.
At leading twist, the forward amplitude admits an expansion in terms of local operators, with the tensor structure once more constrained by Lorentz symmetry, parity, and time-reversal invariance,
\begin{align}
    \tensor{T}{_\mu_\nu}
    &= \sum_{N} \sum_{k \in \set{a,\psq,\psg}} 
     \left( \frac{2}{Q^2}\right)^N 
    \Bigg[\left( \tensor{g}{_\mu_\nu} - \frac{\tensor{q}{_\mu} \tensor{q}{_\nu} }{q^2} \right) \tensor{q}{_{\lambda_1}} \tensor{q}{_{\lambda_2}} \coef{N}{L}[k] (q) \LineNoNumber
    &\qquad + \left( \tensor{g}{_\mu_{\lambda_2}} \tensor{q}{_{\lambda_1}} \tensor{q}{_\nu}  + \tensor{g}{_{\lambda_1}_\nu} \tensor{q}{_\mu} \tensor{q}{_{\lambda_2}} - \tensor{g}{_\mu_\nu} \tensor{g}{_{\lambda_1}_{\lambda_2}} \sp{q}{q} - \tensor{g}{_{\lambda_1}_{\lambda_2}} \tensor{q}{_\mu} \tensor{q}{_\nu}  \right)  \coef{N}{2}[k] (q) \left(\sp{q}{q},m^2 \right) \Bigg] \LineNoNumber
    &\qquad \times  \vectors{q}{_}{3}{N} \bra{p} \tensor{O}{_k^{\{\lambda_1}^\dots^{\lambda_N\}}} \ket{p} 
\, .
\end{align}
Here, braces denote a completely symmetric and traceless combination of indices.
The outer sum runs over all spin values, while the inner sum iterates over all local, symmetric, traceless, twist-two operators of a given spin.
More concretely, these operators are divided into quark non-singlet, quark singlet, and gluon operators
\begin{subequations}
\begin{align}
    \tensor{\mathcal{O}}{_a^{\{\mu_1\dots\mu_N\}}} &= \bar{\psi}_{\light} \lambda^a_\light \tensor{\gamma}{^\{^{\mu_1}} i\tensor{D}{^{\mu_2}} \cdots i\tensor{D}{^{\mu_N}^\}} \psi_{\psq} \comma \LineNumber
    \tensor{\mathcal{O}}{_\psq^{\{\mu_1\dots\mu_N\}}} &= \bar{\psi}_{\light} \tensor{\gamma}{^\{^{\mu_1}} i\tensor{D}{^{\mu_2}} \cdots i\tensor{D}{^{\mu_N}^\}} \psi_{\light} \comma  \LineNumber
    \tensor{\mathcal{O}}{_\psg^{\{\mu_1\dots\mu_N\}}} &= \tensor{G}{^{\{\lambda \mu_1}} i\tensor{D}{^{\mu_2}} \cdots i\tensor{D}{^{\mu_{N-1}}} \tensor{G}{^{\mu_N \lambda \}}} \period
\end{align}
\end{subequations}
The $\lambda^a_l$ denote the generators of the global flavour symmetry group, and $D^{\mu}$ represents the ordinary gauge-covariant derivative.
Since hadrons are described within the parton model, where only massless partonic degrees of freedom are considered, the quark operators are restricted to the light-flavour sector and contributions from heavy-quark operators are omitted at this stage.

Under the assumption that the squared bosonic momentum transfer is much larger than the hadron mass squared, hadronic mass effects in the operator matrix element may be neglected.
In the leading-twist approximation, this gives
\begin{align}
    \bra{p} \tensor{O}{_k^{\{\lambda_1}^\dots^{\lambda_N\}}} \ket{p} 
    &= \vectors[ts]{p}{^}{1}{N} \ome{N}{k} (p) \LineNoNumber
    &= \tensor{p}{^{\phantom{\{} \lambda_{1}}} \dots \tensor{p}{^{\phantom{\{} \lambda_{N}}} \ome{N}{k}(p) + \order{\sp{p}{p}} \, .
\end{align}
These operator matrix elements are not accessible within perturbative QCD, but this is not an obstacle, as the primary interest lies in the coefficients of the operator product expansion of the forward scattering amplitude.
Since the expansion is valid at the operator level, these coefficients are universal and independent of the external states.
Further details are provided in the discussion of renormalisation in the subsequent section.

Upon substituting the operator matrix elements into the operator product expansion of the forward scattering amplitude and performing some basic algebraic manipulations, one finds
\begin{align}
    \tensor{T}{_\mu_\nu}(p,q) &= \sum_{N} \sum_{k \in \set{a,\psq,\psg}} 
    \left( \frac{1}{x} \right)^N 
    \ome{N}{k} (p) \Bigg[ \tensor{e}{_\mu_\nu} \coef{N}{L}[k](q) + \tensor{d}{_\mu_\nu} \coef{N}{2}[k](q)   \Bigg] + \order{p^2} 
\, .
\end{align}
In the leading twist approximation, the forward scattering amplitude may be decomposed in direct analogy to the hadronic tensor, 
\begin{align}
    \tensor{T}{_{\mu\nu}} (p,q) = \tensor{e}{_{\mu\nu}} (p,q) T_L (p,q) + \tensor{d}{_{\mu\nu}} (p,q) T_2 (p,q) + \mathcal{O}(p^2) \period
\end{align}
The hadronic invariants are defined as follows
\begin{subequations}
\begin{align}
    T_k (p,q) &= \tensor{\mathcal{P}}{_{\text{S},k}^{\mu\nu}} (p,q) \tensor{T}{_{\mu\nu}} (p,q) \comma \qquad k\in \{2,L\} \comma
\end{align}
\end{subequations}
where the projectors are given by
\begin{subequations}
\begin{align}
\mathcal{P}_{\mathrm{S},L}^{\mu\nu} (p,q) &= - \frac{\sp{q}{q}}{\sp[2]{p}{q}} \tensor{p}{^\mu}\tensor{p}{^\nu} \comma \LineNumber
\mathcal{P}_{\mathrm{S},2}^{\mu\nu}  (p,q) &= - \frac{(D-1)}{(D-2)}\frac{\sp{q}{q}}{\sp[2]{p}{q}} \tensor{p}{^\mu}\tensor{p}{^\nu} - \frac{1}{(D-2)} \tensor{g}{^\mu^\nu} \period 
\end{align}
\end{subequations}
Having expressed the forward scattering amplitude in terms of the same tensor structures as the hadronic tensor, the hadronic invariants are directly related to the structure functions via the optical theorem,
\begin{align}
    \frac{1}{2\pi} \Im T_i (p,q) &= \frac{1}{2x} F_k (p,q)\comma \qquad i\in \set{2,L} \period
\end{align}
By exploiting the invariance of the forward scattering amplitude for vector–vector scattering under crossing of the external bosons, and by applying the inverse Mellin transform of the operator product expansion together with the optical theorem and crossing symmetry (see Ref.~\cite{Bros:1965kbd}), the hadronic invariants can be related to the structure functions
\begin{align}
     \integral{x}{0}{1} x^{N-2} F_{i} (p,q)
    &= \sum_{k\in\set{a,\psq,\psg}} 
    \ome{N}{k}(p) \coef{N}{i}[k](q)
    \comma \quad i \in \set{2,L} \period
\end{align}
Note that the above relation fixes only the even moments of the structure functions. However, by analytic continuation, all moments in the complex $N$-plane are determined, provided the even moments are known.

\subsection{Renormalisation}

Since the goal is to determine fixed Mellin moments of the coefficient functions from the hadronic invariants, the only remaining complication stems from the hadronic operator matrix elements.
As the operator product expansion of the forward scattering amplitude is valid at the operator level, the coefficient functions are independent of the choice of external states, with all hadronic information encoded solely in the operator matrix elements.

This allows the hadronic states to be replaced by partonic ones carrying momentum $p$, since partons are taken to be massless in contrast to hadrons.
The relevant partonic states consist of non-singlet and singlet quark configurations, as well as gluons, with the corresponding operator matrix elements given by
\begin{align}
    \bra{p,i} \tensor{\mathcal{O}}{_k^{\{\lambda_1}^\dots^{\lambda_N\}}} \ket{p,i} 
    &= \vectors[ts]{p}{^}{1}{N}  \ome{N}{k}[i] (p) \comma \qquad k \in \set{a,\psq,\psg}  \comma
\end{align}
which, in general, are divergent and require renormalisation, see Ref.~\cite{Collins:1984xc} for more details.
To regularise the singularities, dimensional regularisation (see Refs.~\cite{tHooft:1972tcz,Bollini:1972ui,Ashmore:1972uj,Cicuta:1972jf}) is employed in $D = 4 - 2\epsilon$ space–time dimensions.
In this framework, the bare operators are related to the renormalised ones through
\begin{align}
    \tensor{\mathcal{O}}{_k^{\{\lambda_1}^\dots^{\lambda_N\}}}
    &= \sum_{k\in\{\psq,\text{g}\}} Z_{kl} \tensor{\mathcal{O}}{_{k,\ren}^{\{\lambda_1}^\dots^{\lambda_N\}}} \comma \qquad 
    k \in \set{\psq,\psg} \comma
\end{align}
where the renormalisation constants $Z_{kl}$ generally induce mixing among the singlet operators.
The scale dependence of the renormalised operators is then governed by the renormalisation group equation,
\begin{align}
    0 &= \frac{\mathrm{d}}{\mathrm{d}\log \mu^2 } \mathcal{O}_{l}^{\{\lambda_1\dots\lambda_N\}} \LineNoNumber
    &= \dv{\log \mu^2 } \sum_{k\in \set{\psq,\psg}}\left(Z^{-1}\right)_{lk}\mathcal{O}_{k, \ren }^{\{\lambda_1\dots\lambda_N\}} \LineNoNumber
    &= \sum_{k\in \set{\psq,\psg} } \left[ \dv{\left(Z^{-1}\right)_{lk}}{\log \mu^2 } \mathcal{O}_{k, \ren }^{\{\lambda_1\dots\lambda_N\}} +\left(Z^{-1}\right)_{lk}\dv{\mathcal{O}_{k, \ren }^{\{\lambda_1\dots\lambda_N\}}}{\log \mu^2 } \right] \LineNoNumber
    &\equiv -\gamma_{lk} \mathcal{O}_{k, \ren }^{\{\lambda_1\dots\lambda_N\}} +\dv{\mathcal{O}_{l, \ren }^{\{\lambda_1\dots\lambda_N\}}}{\log \mu^2 } \period
\end{align}
Here, the fact that bare operators are independent of the renormalisation scale has been used, together with the relation
\begin{align}
    \delta_{lk} = \sum_{j\in \set{\psg,\psq}} Z_{lj}\left(Z^{-1} \right)_{jk} \; \Rightarrow \; 0 =  \sum_{j\in \set{\psg,\psq}} Z_{lj} \dv{\left(Z^{-1} \right)_{jk}}{\log  \mu^2 } + \sum_{j\in \set{\psg,\psq}} \dv{Z_{lj}}{\log  \mu^2 }\left(Z^{-1} \right)_{jk} \period
\end{align}
The renormalisation constants and the implicitly defined anomalous dimensions $\gamma$ are, in general, matrix–valued. 
Since the $Z_{ij}$ depend on the renormalisation scale $\mu^2$ only implicitly, one finds from the chain rule that
\begin{align}
    \gamma_{lk} = - \sum_{j\in \set{\psg,\psq}} \left(-\epsilon a_{s,\ren} + \beta\right) \pdv{Z_{lj} }{a_{s,\ren}}\left(Z^{-1}\right)_{jk} \comma
\end{align}
with the QCD beta function and the rescaled coupling,
\begin{align}
  a_{s,\ren} = \frac{\alpha_{s,\ren}}{4\pi} \period
\end{align}

All quanitites, the renormalisation constants, the anomalous dimensions, and the beta function admit power–series expansions in the strong coupling.
In addition, the renormalisation constants and anomalous dimensions also expand in the dimensional regulator $\epsilon$.
No positive powers of the regulator are assumed to occur in either expansion, since such terms amount to finite renormalisations.
This assumption is equivalent to working in the modified minimal–subtraction scheme (see Refs.~\cite{Bardeen:1978yd,tHooft:1973mfk,Weinberg:1973xwm}).
Because the anomalous dimensions are finite quantities, this choice of scheme reduces their expansion to a power series in the coupling only.
In this scheme, the expansions read 
\begin{subequations}
\begin{align}
    Z_{ij} &= Z^{(0,0)}_{ij} + \sum_{m=1}^{\infty} \sum_{n=1}^m Z^{(m,n)}_{ij} \frac{a_{s,\ren}^m}{\epsilon^n } \comma \LineNumber
    \gamma_{ij} &= \sum_{m=1}^{\infty} \gamma^{(m)}_{ij} a_{s,\ren}^m\comma \LineNumber
    \beta &= - \sum_{m=2}^{\infty} \beta_{m-2} a_{s,\ren}^m \period
\end{align}
\end{subequations}
For our calculation, we require the Mellin moments of the anomalous dimensions up to two loops (see Refs.~\cite{Gross:1973zrg,Floratos:1977au,Larin:1996wd, Moch:1999eb}) .
Similarly, the coefficients of the beta function $\beta_m$ are needed only up to two loops (see Refs.~\cite{Jones:1974mm,Caswell:1974gg,Egorian:1978zx}).

At zeroth order in the coupling, the operators are finite and do not mix, so no renormalisation is required.
Accordingly, the $Z$–factors may be chosen diagonal and normalised such that
\begin{align}
    Z_{ij} = \delta_{ij} + \order{\alpha_{s,\ren}} \period
\end{align}
With this normalisation, the renormalisation constants are fully determined by the anomalous dimensions and the beta function coefficients, allowing a recursive solution order by order in the strong coupling and the dimensional regulator $\epsilon$. One finds
\begin{subequations}
\begin{align}
    Z_{\text{qq}} &= 1 + a_{s,\ren} \frac{\gamma_{\text{qq}}^{(0)}}{\epsilon} + a_{s,\ren}^2 \Bigg( \frac{-{\beta_0} \gamma_{\text{qq}}^{(0)}+\gamma_{\text{gq}}^{(0)} \gamma_{\text{qg}}^{(0)}+\left(\gamma_{\text{qq}}^{(0)}\right)^2}{2 \epsilon^2}+\frac{\gamma_{\text{qq}}^{(1)}}{2 \epsilon}\Bigg) +\mathcal{O}(a_{s,\ren}^3)\LineNumber
    Z_{\text{qg}} &=  a_{s,\ren} \frac{\gamma_{\text{qg}}^{(0)}}{\epsilon} + a_{s,\ren}^2 \Bigg( \frac{-{\beta_0} \gamma_{\text{qg}}^{(0)}+\gamma_{\text{gg}}^{(0)} \gamma_{\text{qg}}^{(0)}+\gamma_{\text{qq}}^{(0)} \gamma_{\text{qg}}^{(0)}}{2 \epsilon^2}+\frac{\gamma_{\text{qg}}^{(1)}}{2 \epsilon} \Bigg) +\mathcal{O}(a_{s,\ren}^3)\LineNumber
    Z_{\text{gq}} &=  a_{s,\ren} \frac{\gamma_{\text{gq}}^{(0)}}{\epsilon} + a_{s,\ren}^2 \Bigg( \frac{-{\beta_0} \gamma_{\text{gq}}^{(0)}+\gamma_{\text{gg}}^{(0)} \gamma_{\text{gq}}^{(0)}+\gamma_{\text{qq}}^{(0)} \gamma_{\text{gq}}^{(0)}}{2 \epsilon^2}+\frac{\gamma_{\text{gq}}^{(1)}}{2 \epsilon} \Bigg) +\mathcal{O}(a_{s,\ren}^3) \LineNumber
    Z_{\text{gg}} &= 1 +  a_{s,\ren} \frac{\gamma_{\text{gg}}^{(0)}}{\epsilon} + a_{s,\ren}^2 \Bigg( \frac{-{\beta_0} \gamma_{\text{gg}}^{(0)}+\left( \gamma_{\text{gg}}^{(0)} \right)^2+\gamma_{\text{gq}}^{(0)} \gamma_{\text{qg}}^{(0)}}{2 \epsilon^2}+\frac{\gamma_{\text{gg}}^{(1)}}{2 \epsilon} \Bigg) +\mathcal{O}(a_{s,\ren}^3) 
\, .
\end{align}
\end{subequations}
Finally, to renormalise the operator product expansion of the forward amplitude, we replace the bare operators with their renormalised counterparts and, consequently, the operator matrix elements with the renormalised ones.
In the leading-twist approximation, this yields the renormalised partonic invariants,
\begin{align}
    T_{i,\text{p}} (p,q) &= \sum_N  
 \left( \frac{1}{x} \right)^N 
    \sum_{j,k} C_{i,j}^{N} Z_{jk} A^{\text{N,\ren}}_{\text{p},\text{k}} (p) \comma \qquad i \in \set{L,2} \period
\end{align}
This completes the renormalisation of the right-hand side of the forward scattering amplitude, namely its operator product expansion.

To determine the coefficient functions, the forward scattering amplitude must also be computed directly from Feynman diagrams.
These amplitudes require renormalisation as well; throughout the following, the minimal subtraction scheme will be used.

The first step is the renormalisation of the coupling, performed in a theory with $n_\light + n_\heavy$ active flavours, with $a_{s,\ren}^{(n_\light + n_\heavy)}$ denoting the corresponding coupling,
\begin{align}
    \hat{a}_{s,\ren} 
    &= a_{s,\ren}^{(n_\light + n_\heavy)} \left[ 1 - a_{s,\ren}^{(n_\light + n_\heavy)} \frac{1}{\epsilon} \left( \frac{11}{3} C_A - \frac{4}{3} T_F (n_\light + n_\heavy) \right) + \order{(a_{s,\ren}^{(n_\light + n_\heavy)})^2} \right] \period
\end{align}
However, since we perform mass factorisation only with $n_l$ active flavours, this discrepancy must be accounted for by the decoupling relation (see Ref.~\cite{Chetyrkin:1997un}),
\begin{align}
    a_{s,\ren}^{(n_\light + n_\heavy)} &= a_{s,\ren} \bigg\{ 1 + a_{s,\ren} n_\heavy T_F \bigg[ \frac{4}{3} \log \left(\frac{\mu^2}{m^2}\right) + \frac{2}{3} \epsilon \bigg( \log^2 \left(\frac{\mu^2}{m^2}\right) + \zeta_2 \bigg) \LineNoNumber
    &\qquad + \order{\epsilon^2} \bigg]+ \order{a_{s,\ren}^2} \bigg\} \period
\end{align}

For the field renormalisation, it is noted that the vector boson and scalar wave functions require no renormalisation, since QCD corrections to their propagators appear only at next-to-leading order in the respective couplings and are neglected at leading order,
\begin{align}
Z_\vector & = 1 + \order{g_\vector^2} \period
\end{align}

In contrast, gluon and light-quark propagators receive QCD corrections from heavy-quark loops, at one loop for gluons and at two loops for light quarks, making their wave, function renormalisation (see Refs.~\cite{Chen:2025iul,Bierenbaum:2009mv}) necessary.
\begin{subequations}
\begin{align}
    Z_{2} &= 1 + \hat{a}_{s,\ren}^2 \left( \frac{\mu^2}{m^2} \right)^{2\epsilon} C_F T_F n_\heavy \left[\frac{1}{\epsilon} - \frac{5}{6} + \epsilon \left( \frac{89}{36} + \zeta_2 \right) + \mathcal{O}(\epsilon^2) \right] 
    + \mathcal{O}  \left( \hat{a}_{s,\ren}^3 \right) \comma \LineNumber
    Z_{3} &= 1 + \hat{a}_{s,\ren} \left( \frac{\mu^2}{m^2} \right)^{\epsilon} T_F n_\heavy \left[ -\frac{4}{3\epsilon} - \epsilon \frac{2}{3} \zeta_2 + \epsilon^2 \frac{4}{9} \zeta_3 + \mathcal{O}(\epsilon^3) \right] 
    + \mathcal{O}  \left( \hat{a}_{s,\ren}^2 \right) \period
\end{align}
\end{subequations}
Moreover, since internal quarks are allowed to carry a mass, this mass must likewise be renormalised, which is accomplished via counterterm insertions,
\begin{align}
    \delta m = \hat{a}_{s,\ren}   \left( \frac{\mu^2}{m^2} \right)^{\epsilon} \left[ -\frac{3}{\epsilon} -4 - \epsilon \left( 8+\frac{3}{2} \zeta_2 \right) + \mathcal{O}(\epsilon^2) \right] + \mathcal{O}  \left( \hat{a}_{s,\ren} ^2 \right)
\, .
\end{align}

\section{Mellin moments for fixed values of $N$}
\label{sec:hq-mellin-moments}

\subsection{Harmonic Tensors}
We now introduce the harmonic tensors that form the foundation for harmonic projection, which we employ to extract the Mellin moments of the coefficient functions.
Harmonic tensors of rank $N$, constructed from a reference vector $q_{\mu}$ and the metric tensor $g_{\mu\nu}$, are defined as symmetric and traceless tensors, normalised such that
\begin{subequations} \label{eq:H_properties}
\begin{align}
  \tensor{H}{_{q}^{\mu_1}^{\ldots}^{\mu_i}^{\ldots}^{\mu_j}^{\ldots}^{\mu_N}} &= \tensor{H}{_{q}^{\mu_1}^{\ldots}^{\mu_j}^{\ldots}^{\mu_i}^{\ldots}^{\mu_N}} \comma \label{eq:H_symmetry} \LineNumber
  \tensor{g}{_{\mu_i}_{\mu_j}} \tensor{H}{_{q}^{\mu_1}^{\ldots}^{\mu_i}^{\ldots}^{\mu_j}^{\ldots}^{\mu_N}} &= 0 \comma\label{eq:H_trace} \LineNumber
  \tensor{q}{_{\mu_i}} \tensor{H}{_{q}^{\mu_1}^{\ldots}^{\mu_i}^{\ldots}^{\mu_n}} &=  q^2 \tensor{H}{_{q}^{\mu_1}^{\ldots}^{\hat{\mu}_i}^{\ldots}^{\mu_{N}}} \period \label{eq:H_normalisation}
\end{align}
\end{subequations}
Here, the circumflex denotes an omitted index.
With this normalisation, transversality is ensured automatically when the reference vector is light-like.
For an already transverse tensor, the same condition reduces upon Fourier transformation to a harmonic differential equation, establishing a direct link to a spin-$N$ polarisation tensor.

Symmetric, traceless tensors arise throughout physics, from multipole expansions in electrodynamics to gravitational-wave analysis (see Refs.~\cite{Jackson:1998nia, Thorne:1980ru}).
Our normalisation choice, however, differs from standard conventions as the one used in Ref.~\cite{Toth:2021cpx}.
Its purpose and motivation emerge from the explicit construction of a general harmonic tensor.
\begin{align} \label{eq:H_definition}
  \tensor{H}{_{q}^{\mu_1}^{\ldots}^{\mu_N}} = \sum_{K=0}^{\lfloor\frac{N}{2}\rfloor} (q^2)^{K} h^{N}_{2K} \sum_{\sigma / \sim} \tensor{g}{^{\sigma(\mu_1)}^{\sigma(\mu_2)}} \ldots \tensor{g}{^{\sigma(\mu_{2K-1})}^{\sigma(\mu_{2K})}} \tensor{q}{^{\sigma(\mu_{2K+1})}} \ldots \tensor{q}{^{\sigma(\mu_{N})}} \comma
\end{align}
where the sum runs over all permutations $\sigma$ of the indices, modulo those that are equivalent either by the symmetry of the metric tensor or by the commutativity of products, and the corresponding coefficients are defined by
\begin{align} \label{eq:H_coefficient}
  h^{N}_{2K} = (-1)^K 2^{N} \frac{ \gammaf{D-2} \gammaf{D/2-1+N-K}}{ \gammaf{D-2+N}\gammaf{D/2-1} } \period
\end{align}
This definition of the coefficients accounts for redundancies arising from the intrinsic symmetries of the metric tensor $\tensor{g}{_\mu_\nu}$ and products of the reference vector $\tensor{q}{_\mu}$.
Consequently, we sum only over permutations distinct under these equivalences, absorbing the combinatorial factors into the coefficients.

\subsection{Harmonic Projection}
To efficiently compute fixed Mellin moments of the coefficient functions, the method of projectors (see Refs.~\cite{ Gorishnii:1983su, Gorishnii:1986gn}) is used to project on individual terms in the operator product expansion. 
The harmonic projection operator is defined as 
\begin{align}
    \mathcal{P}_M(\cdot) = \frac{1}{2^M} \frac{\tensor{H}{_q^{\mu_1}^\ldots^{\mu_M}}}{M!} \frac{\partial^M}{\partial \tensor{p}{^{\mu_1}} \ldots  \partial \tensor{p}{^{\mu_M}}} \Bigg( \cdot \Bigg) \Bigg|_{p=0} \period
\end{align}
It has the property of yielding a non-zero result only when applied to a product of exactly $M$ vectors $p$.
To see this, consider $M \neq N$ and evaluate the action of $\mathcal{P}_M$ explicitly. Making use of the symmetry of the product and the defining properties of harmonic tensors, the indices can be relabelled as needed, leading to the derivative
\begin{align}
   \mathcal{P}_M (\vectors{p}{^}{1}{N}) &= \begin{cases}
       \frac{1}{2^M}  \frac{N!}{M!} \tensor{H}{_q^{\mu_1}^\ldots^{\mu_M}} \tensor{\delta}{^{\lambda_1}_{\mu_1}} \dots \tensor{\delta}{^{\lambda_M}_{\mu_M}} \vectors{\partial}{^}{M+1}{N} \Big|_{p=0} & \text{if } N < M  \\
       \frac{1}{2^M}  \frac{N!}{M!} \tensor{H}{_q^{\mu_1}^\ldots^{\mu_M}} \tensor{\delta}{^{\lambda_1}_{\mu_1}} \dots \tensor{\delta}{^{\lambda_M}_{\mu_M}} \vectors{p}{^}{M+1}{N} \Big|_{p=0} & \text{if } N > M \period
   \end{cases}
\end{align}
In the first case, the term vanishes trivially, since having more derivatives than factors reduces to a derivative of a constant, i.e. zero. In the second case, the projection vanishes as well, because at least one factor of $p$ remains. Thus, the projection is non-zero only when $M = N$; in particular, one has
\begin{align}
   \mathcal{P}_N (\vectors{p}{^}{1}{N}) 
   &= \frac{1}{2^M} \tensor{H}{_q^{\lambda_1}^\ldots^{\lambda_N}} 
\, ,
\end{align}
so that the result itself is a harmonic tensor. 
Consequently, the harmonic projection applied to products of momenta $p$ yields a harmonic tensor. In practice, the main interest lies in its application to the forward scattering amplitude. On the one hand, this amounts to applying it to the operator product expansion, i.e., to the hadronic invariants, where expressions of the form 
\begin{align}
    \mathcal{P}_M(T_{i}) 
    &= (-1)^M \sum_k \ome{N}{k} (p) \coef{N}{i}[k](q)\Bigg|_{p=0}
\end{align}
arise. 
Here it is used that by symmetry considerations the operator matrix elements depend only on $p^2$; thus, by the chain rule, derivatives acting on them do not reduce the total number of factors of $p$, so that all terms with derivatives acting on the operator matrix element vanish.

On the other hand, when evaluating the forward scattering amplitude through Feynman diagrams, the harmonic projection acts on scalar products involving the partonic momentum, especially those originating from the numerator structure of the Feynman rules. Further scalar products can appear in the denominators (see Ref.~\cite{Witten:1975bh}), as these may be expanded in powers of $p$, for example:
\begin{align*}
    \frac{1}{(p+q)^2-m^2} 
    &=  \frac{1}{\sp{q}{q}-m^2} \sum_{n=0}^\infty (-1)^n \left(\frac{2\sp{p}{q}}{\sp{q}{q}-m^2} \right)^n
\, .
\end{align*}
In this case, the application of the harmonic projection is conceptually straightforward, as it amounts to differentiating the integrands and contracting harmonic tensors.

The expansion, however, introduces the first bottleneck of the method when higher moments are considered. By the Leibniz rule, the number of generated terms grows with the number of edges carrying the momentum $p$, so for practical calculations this number should be minimised.

\subsection{Mass Factorisation}
With all necessary ingredients assembled, the calculation can now be set up.
The harmonic projection is solely concerned with the $p$–dependence of the operator matrix elements, without regard to the specific processes they describe.
Moreover, it commutes with renormalisation, since the counterterms are local and therefore independent of $p$.
Consequently, the result obtained for the harmonic projection of the hadronic invariants carries over directly to the renormalised partonic invariants, ensuring that the projection can be performed prior to specifying the external states,
\begin{align}
T_{i,\text{p}}^{\ren} (p,q) &= (-1)^N \sum_{j,k \in \set{a,\psq,\psg}}  \coef{N}{i}[j](q) Z_{jk} \ome{N,\ren}{\parton}[k](p) \bigg\vert_{p=0} \comma 
\qquad
i \in \set{L,2}, \parton \in \set{\psg,\psq} \period
\end{align}
Nullifying the partonic momentum affects only the operator matrix element.
In a theory with light flavours only, this matrix element would contribute solely at tree level after nullification of $p$, since all higher order terms in the coupling correspond to scaleless loop integrals, which vanish in dimensional regularisation.
The situation changes in the presence of heavy flavours as the heavy-quark mass provides a physical scale, giving rise to non-vanishing massive bubble integrals,
\begin{align}
 \ome{N,\ren}{\parton}[k] (p) \bigg\vert_{p=0} &=  \ome{N,\ren,(0)}{\parton}[k] (p)  \bigg\vert_{p=0} + a_{s,\ren} \ome{N,\ren,(1)}{\parton}[k] (p)  \bigg\vert_{p=0}  + a_{s,\ren} ^2 \ome{N,\ren,(2)}{\parton}[k] (p) \bigg\vert_{p=0}  \LineNoNumber 
 &\qquad + a_{s,\ren} ^3 \ome{N,\ren,(3)}{\parton}[k] (p) \bigg\vert_{p=0} + \order{a_{s,\ren} ^4} \period
\end{align}
This issue is resolved by evaluating the operator matrix elements in a kinematic renormalisation scheme (see Ref.~\cite{Celmaster:1979km}) in which the finite contributions from the massive bubble integrals are removed through a finite renormalisation.
In such a scheme,
\begin{align}
  \ome{N,\ren}{\parton}[k] (p) \bigg\vert_{p=0} &= \ome{N,\ren,(0)}{\parton}[k] (p)  \bigg\vert_{p=0} \delta_{\text{p,k}} \LineNoNumber
  &\equiv \const_k \delta_{\text{p,k}} \comma
\end{align}
which is constant with respect to the external momentum, though it still depends on $\epsilon$ and the Mellin moment $N$.
The Kronecker delta arises because the off–diagonal operator matrix elements have no tree–level contributions and therefore vanish when the partonic momentum is set to zero.
Depicting the inserted operator vertex as a circle labelled by $N$, which denotes the spin of the operator, the operator matrix element may be represented graphically as
\begin{subequations}
\begin{align} 
   \ome{N,\ren}{\psq}[\psq] (p) \Bigg\vert_{p=0} &= \vcenter{\hbox{\includegraphics[width=3cm]{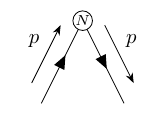}}}\Bigg\vert_{p=0} + \vcenter{\hbox{\includegraphics[width=3cm]{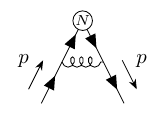}}}\Bigg\vert_{p=0} + \mathcal{O}(\alpha_s^2) \LineNoNumber
    &\equiv k^N_\psq \comma
    \LineNumber
     \ome{N,\ren}{\psq}[\psg] (p) \Bigg\vert_{p=0}  &= \vcenter{\hbox{\includegraphics[width=3cm]{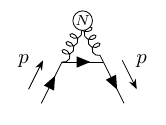}}}\Bigg\vert_{p=0} + \mathcal{O}(\alpha_s^2) \LineNoNumber
    &= 0 \comma \LineNumber
         \ome{N,\ren}{\psg}[\psg] (p) \Bigg\vert_{p=0} &= \vcenter{\hbox{\includegraphics[width=3cm]{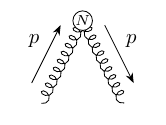}}}\Bigg\vert_{p=0} + \vcenter{\hbox{\includegraphics[width=3cm]{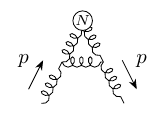}}}\Bigg\vert_{p=0} + \mathcal{O}(\alpha_s^2) \LineNoNumber
    &\equiv \text{k}^{N}_\psg \comma
    \LineNumber
     \ome{N,\ren}{\psg}[\psq] (p) \Bigg\vert_{p=0}  &= \vcenter{\hbox{\includegraphics[width=3cm]{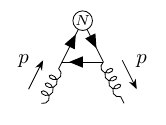}}}\Bigg\vert_{p=0} + \mathcal{O}(\alpha_s^2) \LineNoNumber
    &= 0 \comma
\end{align}
\end{subequations}
where the $k^N_\parton$ are process dependent constants, i.e., they may depend on the regulator $\epsilon$ due to the polarisation sum for the external gluons. 
Substituting the operator matrix elements, one finds the renormalised partonic invariants
\begin{align}
   T_{i,\parton}^{N,\ren} (q) &= (-1)^N \sum_{j \in \set{a,\psq,\psg}} \coef{N}{i}[j](q) Z_{j \parton} \ome{N,\ren,(0)}{\parton}[\parton](p) \bigg\vert_{p=0} \comma  \quad i \in \set{L,2} \period
\end{align}

To extract the finite coefficient functions directly, assume that the coefficient function admits an expansion in the coupling and the regulator $\epsilon$, such that
\begin{align}
\label{eq:def-coeffs}
    \coef{N}{i}[j](q) &= \sum_{m=0}^{\infty} \sum_{n=0}^{\infty} \cmom{i}{\parton}{m}{n} (q)a_{s,\ren}^m \epsilon^n 
\, .
\end{align}
Upon substituting the series expansions of the coefficient functions together with the $Z$–factors derived earlier, the normalisation is fixed at tree level,
\begin{subequations}
\begin{align}
T_{2,\psq}^{N,\ren} (q) &= 1 + \order{a_{s,\ren}} \comma \LineNumber
T_{2,\psg}^{N,\ren} (q) &= \order{a_{s,\ren}} \comma \LineNumber
T_{L,\parton}^{N,\ren} (q)&=  \order{a_{s,\ren}} , \qquad \parton \in \set{\psq, \psg}  \period
\end{align}
\end{subequations}
Note that this normalisation implicitly fixes the tree–level contributions to the moments of the coefficient functions.
Substituting the series expansions of the renormalisation constants and the coefficient functions, one finds
\begin{subequations}
\begin{align}
 \label{eq:renA}
& (-1)^N  k^{N}_\psq T_{2,\psq}^{N,\ren} = 1 + a_{s,\ren} \Big\{ \epsilon^{-1} \adim{\psq}{\psq}{0}+\cmom{2}{\psq}{1}{0}+\epsilon \cmom{2}{\psq}{1}{0} +\order{\epsilon^2}\Big\} \LineNoNumber 
  &\qquad + a_{s,\ren}^2 \Big\{\epsilon^{-2} \Big[ \tfrac{1}{2} \big(\adim{\psq}{\psg}{0} \adim{\psg}{\psq}{0} + \big(\adim{\psq}{\psq}{0}\big)^2\big) -\adim{\psq}{\psq}{0} \bfun{0} \Big] + \epsilon^{-1} \Big[ \tfrac{1}{2}\adim{\psq}{\psq}{1} + \cmom{2}{\psq}{1}{0} \adim{\psq}{\psq}{0} \phantom{\Big[}\LineNoNumber 
  &\qquad + \cmom{2}{\psg}{1}{0} \adim{\psg}{\psq}{0}\Big] + \cmom{2}{\psq}{1}{1} \adim{\psq}{\psq}{0} + \cmom{2}{\psq}{2}{0} + \cmom{2}{\psg}{1}{1} \adim{\psg}{\psq}{0} + \order{\epsilon} \Big\} + \order{a_{s,\ren}^3} \LineNumber
&  (-1)^N k^{N}_\psg T_{i,\psg}^{N,\ren} = a_{s,\ren} \Big\{ \epsilon^{-1}\adim{\psq}{\psg}{0}+\cmom{2}{\psg}{1}{0}+\epsilon\cmom{2}{\psg}{1}{0}+ +\order{\epsilon^2}\Big\} \LineNoNumber 
  &\qquad +a_{s,\ren}^2\Big\{\epsilon^{-2} \Big[ \tfrac{1}{2} \big(\adim{\psg}{\psg}{0} \adim{\psq}{\psg}{0} + \adim{\psq}{\psg}{0} \adim{\psq}{\psq}{0} - \adim{\psq}{\psg}{0} \bfun{0}\big) + \epsilon^{-1} \Big[ \tfrac{1}{2} \adim{\psq}{\psg}{1}   \phantom{\Big[}  + \cmom{2}{\psg}{1}{0} \adim{\psg}{\psg}{0} \LineNoNumber
    &\qquad + \cmom{2}{\psq}{1}{0} \adim{\psq}{\psg}{0}\Big]+ \cmom{2}{\psg}{1}{1} \adim{\psg}{\psg}{0} + \cmom{2}{\psg}{2}{0} + \cmom{2}{\psq}{1}{1} \adim{\psq}{\psg}{0} +  \order{\epsilon} \Big\} + \order{a_{s,\ren}^3} \LineNumber
& (-1)^N k^{N}_\psq T_{L,\psq}^{N,\ren} = a_{s,\ren} \Big\{ \cmom{L}{\psq}{1}{0}+\epsilon\cmom{L}{\psq}{1}{0}+\order{\epsilon^2}\Big\}  +a_{s,\ren}^2 \Big\{ \epsilon^{-1} \Big[ \cmom{L}{\psq}{1}{0} \adim{\psq}{\psq}{0} \LineNoNumber \phantom{\Big[}
    &\qquad + \cmom{L}{\psg}{1}{0} \adim{\psg}{\psq}{0} \Big]+ \cmom{L}{\psq}{1}{1} \adim{\psq}{\psq}{0} + \cmom{L}{\psq}{2}{0} + \cmom{L}{\psg}{1}{1} \adim{\psg}{\psq}{0}
  + \order{\epsilon} \Big\} + \order{a_{s,\ren}^3} \LineNumber
& (-1)^N k^{N}_\psg T_{L,\psg}^{N,\ren} = a_{s,\ren} \Big\{ \cmom{L}{\psg}{1}{0}+\epsilon\cmom{L}{\psg}{1}{0}+\order{\epsilon^2}\Big\} +a_{s,\ren}^2 \Big\{ \epsilon^{-1} \Big[ \cmom{L}{\psg}{1}{0} \adim{\psg}{\psg}{0} \LineNoNumber \phantom{\Big[}
    &\qquad + \cmom{L}{\psq}{1}{0} \adim{\psq}{\psg}{0} \Big]+ \cmom{L}{\psg}{1}{1} \adim{\psg}{\psg}{0} + \cmom{L}{\psg}{2}{0} + \cmom{L}{\psq}{1}{1} \adim{\psq}{\psg}{0} + \order{\epsilon} \Big\} + \order{a_{s,\ren}^3}
    \label{eq:renB}
\, .
\end{align}
\end{subequations}
The constants on the left-hand side are determined by the tree-level operator matrix elements, which are assumed to be normalised such that
\begin{subequations}
\begin{align}
k^N_\psq &= 1 \comma \LineNumber 
k^N_\psg &= (1-\epsilon)^{-1} \period
\end{align}
\end{subequations}
From this point, the finite part of the coefficient functions can be extracted recursively, once the left-hand side has been computed.

\subsection{Details of the Calculation}
When calculating the forward scattering amplitude, we sum over all gluon polarisations and light-quark spins and average over all colours.
For the polarisation sum, we could in principle employ a physical gauge with
\begin{align}
       \sum_{s \in \mathcal{S}_\psg} \tensor{\epsilon}{_{s,\mu}} (p)\tensor{\epsilon}{^*_{s,\nu}} (p)= -\tensor{g}{_{\mu\nu}} + \frac{\tensor{p}{_\mu}\tensor{q}{_\nu} +\tensor{p}{_\nu}\tensor{q}{_\mu}}{p \cdot q} - \frac{q^2 \tensor{p}{_\nu}\tensor{p}{_\mu}}{(p \cdot q)^2} \comma
\end{align}
since $q$ is not collinear to $p$.
Here, $\mathcal{S}_\psg$ denotes the set of gluon polarisations and $\tensor{\epsilon}{_{s,\nu}}$ the polarisation tensor with respect to the momentum $p$.
If we employ the physical polarisation sum, the harmonic projection, being essentially a derivative operator, will generate numerous additional terms through extra factors of the momentum $p$.
Instead, for computational efficiency, we prefer an unphysical gauge, employing the polarisation sum
\begin{align}
    \sum_{s \in \mathcal{S}_\psg}  \tensor{\epsilon}{_{s,\mu}} (p)\tensor{\epsilon}{^*_{s,\nu}} (p)= -\tensor{g}{_{\mu\nu}} \period
\end{align}
This approach reduces the number of generated terms, but requires ghosts as external partonic states to compensate for the unphysical gluonic degrees of freedom.
Internal gluons are taken in a general $R_\xi$ gauge, with gauge invariance ensuring that the gauge parameter $\xi$ cancels in the final result, thereby providing a consistency check.

Using these conventions, the  partonic invariants can be decomposed as
\begin{subequations}
\begin{align}
    T_{i,\psq} (p,q) &= T^{\psq \vector \psq \vector}_{i} (p,q) \comma  &&  i \in \set{L,2} \LineNumber
    T_{i,\psg} (p,q)&= T^{\psg \vector \psg \vector}_{i} (p,q)+ 2 T^{\psc \vector \psc \vector}_{i}(p,q)\comma &&  i \in \set{L,2} \period
\end{align}
\end{subequations}
The factor of two for the ghost contributions is required, since both ghosts and anti-ghosts in the initial states must be summed over.
Superscripts denote the external particle content contributing to the forward scattering amplitude.

As defined previously, the projectors are not optimal for efficient calculation, as they are redundant in the sense that the longitudinal part is evaluated twice.
A further shortcoming is their dependence on the dimensional regulator $\epsilon$, which unnecessarily complicates intermediate steps of the calculation.
To address these issues, we decompose the projectors as
\begin{subequations}
\begin{align}
    \tensor{\mathcal{P}}{_{\text{S,}}_L^\mu^\nu} &= -\frac{1}{\sp{q}{q}} \tensor{\hat{\mathcal{P}}}{_2^\mu^\nu} \comma \LineNumber
    \tensor{\mathcal{P}}{_{\text{S,}}_2^\mu^\nu} &= -\frac{1}{\sp{q}{q}} \frac{3-2\epsilon}{2-2\epsilon} \tensor{\hat{\mathcal{P}}}{_2^\mu^\nu}  - \frac{1}{2-2\epsilon}  \tensor{\hat{\mathcal{P}}}{_0^\mu^\nu} \comma \LineNumber
\end{align}
\end{subequations}
with the simplified projectors
\begin{subequations}
\begin{align}
    \tensor{\hat{\mathcal{P}}}{_0^\mu^\nu} 
    &= \tensor{g}{^\mu^\nu} \comma \LineNumber 
    \tensor{\hat{\mathcal{P}}}{_2^\mu^\nu} 
    &= 4 x^2 \tensor{p}{^\mu}\tensor{p}{^\nu} \period
\end{align}
\end{subequations}
These projectors not only avoid redundant calculations but also allow a more efficient and modular implementation.
The reason is that they disentangle the Mellin moments, allowing the forward amplitude to be expanded to a fixed power of $p$ rather than requiring term-by-term separation according to powers of the partonic momentum. 

It should be noted that the flavour factors, namely the charges of the vector bosons, decouple from the amplitudes and appear as multiplicative factors.
This allows for the simultaneous calculation of the non-singlet and singlet contributions, as the amplitudes modulo the flavour factors are identical.
The flavour factors are fully determined by the open and closed fermion lines appearing in a graph and the vector bosons coupling to them.
Therefore, to define the factors, it is sufficient to consider only a single graph representing each such configuration.
This is illustrated in \autoref{fig:flavour_graphs}, where, for completeness, all flavour factors that may appear, including those at three loops, are displayed for pure singlet configurations in the presence of a neutral current.
The concrete definitions are provided in \autoref{app:flavourfactors}, since the flavour factors are not relevant for the subsequent discussion and become important only in phenomenological applications.
\begin{figure}[H]
\centering
\fbox{%
  \begin{minipage}{0.92\textwidth}  
    \centering
    \begin{tabular}{ccc}
      \includegraphics[width=0.25\textwidth]{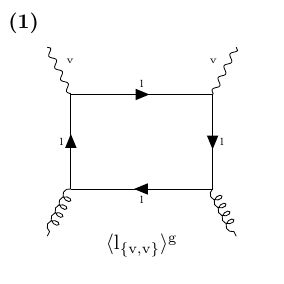} &
      \includegraphics[width=0.25\textwidth]{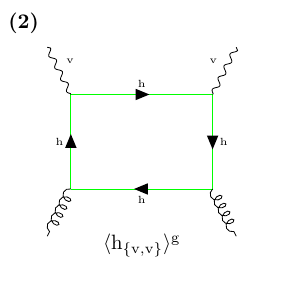} &
      \includegraphics[width=0.25\textwidth]{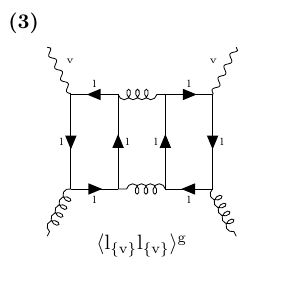} \\[0.25cm]
      \includegraphics[width=0.25\textwidth]{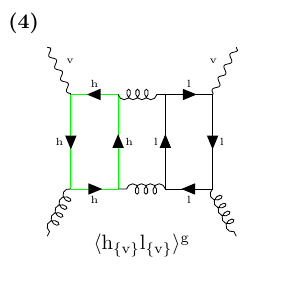} &
      \includegraphics[width=0.25\textwidth]{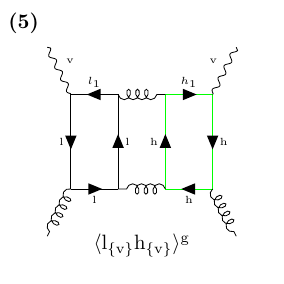} &
      \includegraphics[width=0.25\textwidth]{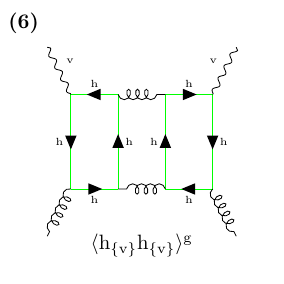} \\[0.25cm]
      \includegraphics[width=0.25\textwidth]{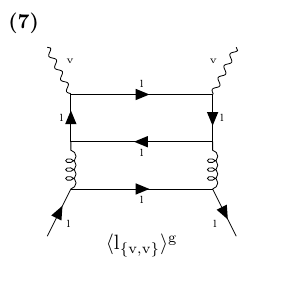} &
      \includegraphics[width=0.25\textwidth]{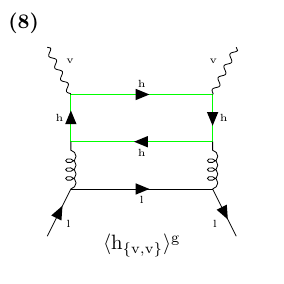} &
      \includegraphics[width=0.25\textwidth]{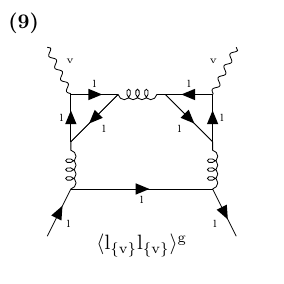} \\[0.25cm]
      \includegraphics[width=0.25\textwidth]{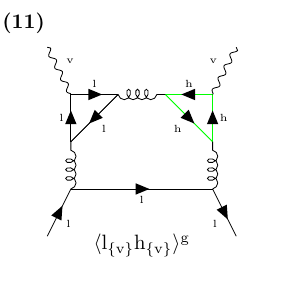} &
      \includegraphics[width=0.25\textwidth]{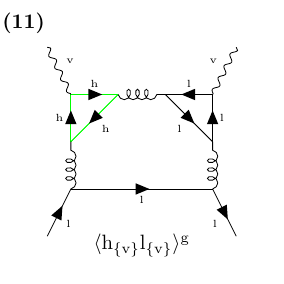} &
      \includegraphics[width=0.25\textwidth]{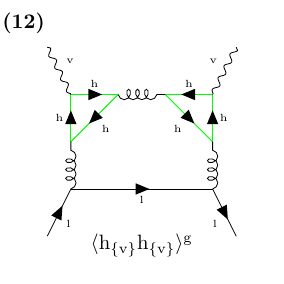} \\[0.25cm]
      \includegraphics[width=0.25\textwidth]{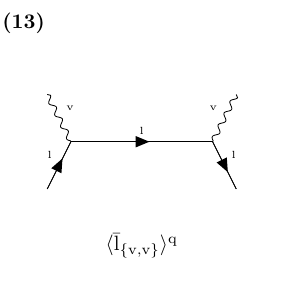} &
      \includegraphics[width=0.25\textwidth]{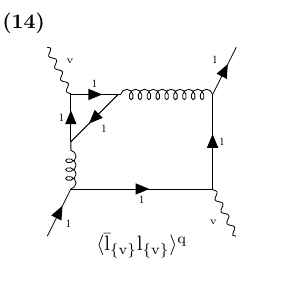} &
      \includegraphics[width=0.25\textwidth]{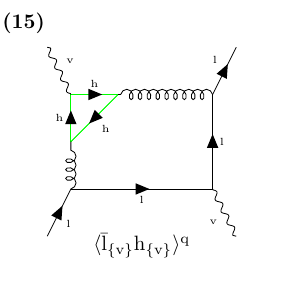}
    \end{tabular}
  \end{minipage}
}
\caption{Collection of Feynman graphs showing examples for all flavour structures that appear in neutral current forward scattering amplitudes. Light quarks $\light$ are depicted in black, heavy quarks $\heavy$ in green, and vector bosons by $\vector$.}
\label{fig:flavour_graphs}
\end{figure}

\subsection{Feynman Graphs}
The first step in computing the forward scattering amplitude is to generate all graphs contributing to the processes of interest, excluding tadpoles but including snails,
\begin{align}
\label{eq:processes}
    \parton + \vector &\longrightarrow \parton + \vector , &&\quad \parton \in \set{\psq, \psg, \psc}
\end{align}
The graph generation was performed using the generator provided with \texttt{FORM} (see Refs.~\cite{Kaneko:1994fd,Davies:2026cci}), 
and the results were checked independently with \texttt{QGRAF} (see Ref.~\cite{Nogueira:1991ex}). 
The so-generated graphs are not free of redundancies. 
To remove them, it is advantageous to work with equivalence classes of graphs rather than individual graphs, declaring two graphs equivalent if their moments yield the same amplitude up to an overall sign. 
These redundancies are eliminated in several stages by exploiting pattern matching and symmetry properties, following the procedures in \texttt{FORM} (see Refs.~\cite{Vermaseren:2000nd,Kuipers:2012rf,Ruijl:2017dtg,FORM_v4.3.1}).

As a preliminary step, all graphs that do not contribute to the amplitude must be identified, as they are all equivalent to zero.
The first such case involves graphs with vanishing colour factors.
For this purpose, it is advantageous to decompose all four gluon vertices, which separates the colour structures from the remaining numerator structures of the amplitude,
\begin{align} 
    \vcenter{\hbox{\includegraphics[width=3.7cm]{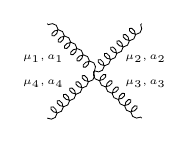}}} 
    &=  
    \vcenter{\hbox{\includegraphics[width=3.7cm]{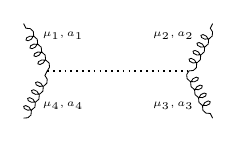}}} 
    +
    \vcenter{\hbox{\includegraphics[width=3.7cm]{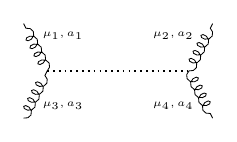}}} 
    +
    \vcenter{\hbox{\includegraphics[width=3.7cm]{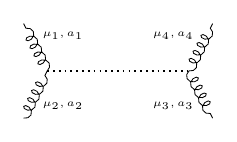}}} 
\end{align}
Once all four-gluon vertices are decomposed into three-valent ones, the colour structure effectively decouples from the Lorentz structure of the amplitudes.
At this point, the Feynman rules can be assigned to the graphs, and the \texttt{FORM} package \texttt{color.h} (see Ref.~\cite{vanRitbergen:1998pn}) can then be employed to compute all colour traces.

In the next step, other graphs contributing to the equivalence class of trivial graphs are determined.
This is done by identifying graphs that contain scaleless subgraphs after the harmonic projection is applied.
The reason is that the nullification of the partonic momentum effectively amputates the partonic edges whilst keeping the corresponding vertices.
The resulting graphs may then contain scaleless subgraphs upon identifying edges joined by a two-valent vertex that vanish in dimensional regularisation.
The criterion for identifying scaleless subgraphs is as follows.
After performing the harmonic projection, that is graphically the amputation of the partonic edges and contracting the edges joined by two-valent vertices, the resulting reduced graphs are examined.
They are analysed to determine whether they contain massless self-loops or are one-vertex reducible.
In the latter case, a graph that is not one-vertex reducible is considered scaleless if, upon the removal of a single vertex, it separates into more than one disjoint component, one of which contains neither external edges nor massive lines.
Another contribution to trivial graphs arises from cancellations due to charge conjugation, namely Furry's theorem.
These graphs may be detected directly at the level of graph generation.

The final step is to relate graphs differing only by crossing external particles.
In forward scattering amplitudes, such a crossing corresponds to reversing the momenta $p$ and $q$ and, where relevant, exchanging indices between two particles of the same type.
To examine the effect of reversing the external momenta, consider a graph with $L$ loops and only external scalar particles. Then, consider the forward scattering amplitude $\mathcal{M}_\Gamma$ with integrand $\mathcal{I}_\Gamma$, and project onto the $N$-th moment using the harmonic projection operator 
\begin{align}
\mathcal{P}_N(\mathcal{M}(p,q)) &\equiv 
\intn{k_1}{D} \dots \intn{k_L}{D} \mathcal{P}_N(\mathcal{I}(p,q)) \LineNoNumber
    &= \intn{k_1}{D} \dots \intn{k_L}{D}\sum_{M=0}^\infty   \mathcal{P}_N  \left( \tensor{p}{_{\mu_1}} \dots \tensor{p}{_{\mu_M}} \tensor{\hat{\mathcal{I}}}{^{\mu_1}^\dots^{\mu_M}}(q) \right) \LineNoNumber
    &= \intn{k_1}{D} \dots \intn{k_L}{D}  \mathcal{P}_N \left( \tensor{p}{_{\mu_1}} \dots \tensor{p}{_{\mu_N}} \tensor{\hat{\mathcal{I}_\Gamma}}{^{\mu_1}^\dots^{\mu_M}}(q) \right) \period
\end{align}
By Lorentz symmetry, $\hat{\mathcal{I}}$ can be constructed only from the vector $q$ and the metric tensor $g^{\mu\nu}$, and must be totally symmetric to be compatible with contraction with a totally symmetric tensor.
Since the metric tensor carries two indices, tensors with an even number of indices must contain an even power of the momentum $q$, whereas those with an odd number of indices must contain an odd power.
Furthermore, all coefficients may only depend on $q^{2}$, so that
\begin{align}
 \intn{k_1}{D} \dots \intn{k_L}{D}  \tensor{\hat{\mathcal{I}}}{_\Gamma^{\mu_1}^\dots^{\mu_M}}(q) &= \sum_{k=0}^{\lfloor\frac{N}{2}\rfloor} \hat{\mathcal{M}}^{N}_{\Gamma,2k} (q^2) \LineNoNumber
 &\quad \times \sum_{\sigma / \sim} \tensor{g}{^{\sigma(\mu_1)}^{\sigma(\mu_2)}} \ldots \tensor{g}{^{\sigma(\mu_{2k-1})}^{\sigma(\mu_{2k})}} \tensor{q}{^{\sigma(\mu_{2k+1})}} \ldots \tensor{q}{^{\sigma(\mu_{n})}}  \comma
\end{align}
where the sum runs over all independent permutations $\sigma$ of the indices, as described in the definition of harmonic tensors.
Hence, the $N$-th moment of the amplitude becomes
\begin{align}
\mathcal{P}_N(\mathcal{M}_\Gamma(p,q)) 
    &= \mathcal{P}_N \Bigg( \tensor{p}{_{\mu_1}} \dots \tensor{p}{_{\mu_N}} \sum_{k=0}^{\lfloor\frac{N}{2}\rfloor} \hat{\mathcal{M}}^{N}_{\Gamma,2k} (q^2) \LineNoNumber
    &\quad \times \sum_{\sigma / \sim} \tensor{g}{^{\sigma(\mu_1)}^{\sigma(\mu_2)}} \ldots \tensor{g}{^{\sigma(\mu_{2k-1})}^{\sigma(\mu_{2k})}} \tensor{q}{^{\sigma(\mu_{2k+1})}} \ldots \tensor{q}{^{\sigma(\mu_{n})}} \Bigg) \period
\end{align}
Consequently, the amplitude is an odd function of either momentum for odd $N$ and an even function for even $N$.
This property applies only to the Mellin moments, not to the amplitudes themselves, which in general are neither purely odd nor purely even functions of the momenta.

To extend this symmetry to individual graphs, one must note that the above argument applies only to graphs with external scalar edges.
Consider now the case where the external edges carrying momentum $p$ correspond to partons.
For gluons, exchanging one gluon for another is always possible since they are their own antiparticles.
The polarisation sum is manifestly symmetric in the gluon indices, requiring that momenta and their associated indices to be swapped together.
External quarks and ghosts differ, as reversing quark lines is generally not possible, with similar restrictions applying to ghost lines.
These parity transformations are therefore admissible only for gluons, leading to the graphical identities,
\begin{align}
    \mathcal{P}_N \left(
    \vcenter{\hbox{\scalebox{0.6}{
    \includegraphics{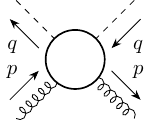}
    }}}
    \right)
    &= (-1)^N \mathcal{P}_N \left(
    \vcenter{\hbox{\scalebox{0.6}{
    \includegraphics{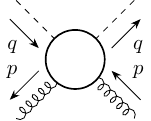}
    }}}
    \right)
    = (-1)^{N} \mathcal{P}_N \left(
    \vcenter{\hbox{\scalebox{0.6}{
    \includegraphics{images/ggHpm.pdf}
    }}}
    \right)
    = \mathcal{P}_N \left(
    \vcenter{\hbox{\scalebox{0.6}{
    \includegraphics{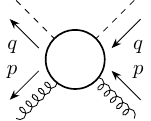}
    }}}
    \right) \period
\end{align}
If the external boson is taken to be a vector boson rather than a scalar, it is contracted with a symmetric tensor when projecting onto the partonic invariants.
Exchanging the bosonic momenta together with their indices therefore does not introduce an additional sign,
\begin{align}
    \mathcal{P}_N \left(
    \vcenter{\hbox{\scalebox{0.6}{
    \includegraphics{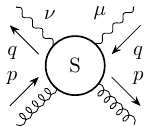}
    }}}
    \right)
    &= (-1)^N \mathcal{P}_N \left(
    \vcenter{\hbox{\scalebox{0.6}{
    \includegraphics{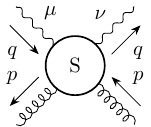}
    }}}
    \right)
    = (-1)^{N} \mathcal{P}_N \left(
    \vcenter{\hbox{\scalebox{0.6}{
    \includegraphics{images/ggSpm.pdf}
    }}}
    \right)
    = \phantom{-}\mathcal{P}_N \left(
    \vcenter{\hbox{\scalebox{0.6}{
    \includegraphics{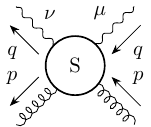}
    }}}
    \right) \comma 
\end{align}
where, for these symmetries to make sense, both vector bosons must be of the same particle type.

In the end, each diagram corresponds to an equivalence class.
To identify equivalences up to possible signs, the non-trivial graphs contributing to the calculation are matched against these equivalence classes, allowing for crossing of external particles.
Any signs arising in the process are accounted for directly.
Finally, the graphs are grouped into super-graphs (see Ref.~\cite{Ruijl:2017cxj}), defined as sums over all amplitudes with identical colour and flavour factors, up to overall numerical coefficients, that belong to the same integral family. 
The numbers of graphs up to two loops for the processes considered in Eq.~(\ref{eq:processes}) 
are listed in \autoref{tab:number-of-graphs}.

But it should be emphasised that the application of the symmetries is not crucial for the calculation of the fixed-order moments.
The reason is that dressing the graphs with Feynman rules is essentially constant and independent of the order of the moment to be calculated.
However, utilising the symmetries does nonetheless significantly improve the calculation at intermediate steps and provides a non-trivial self-consistency check, because the calculated amplitude should respect the symmetries.
\newcolumntype{Y}{>{\centering\arraybackslash}X} 
\begin{table}[ht]
\centering
\begin{tabularx}{\textwidth}{|Y|Y|Y|Y|Y|Y|Y|}
\hline
\textbf{process} & \multicolumn{2}{c|}{\textbf{unsymmetrised}} & \multicolumn{2}{c|}{\textbf{symmetrised}} & \multicolumn{2}{c|}{\textbf{super}} \\
\hline
qvqv & 12 & 234 & 3 & 40  & 1 & 7 \\
cvcv & 0  & 78  & 0 & 4   & 0 & 3 \\
gvgv & 16 & 352 & 4 & 34  & 2 & 6 \\
\hline
\end{tabularx}
\caption{
\label{tab:number-of-graphs}
Number of graphs contributing to the forward scattering amplitude at one-loop and two-loop level for each process. For each category unsymmetrised, symmetrised, and super graphs, the first column shows the one-loop count and the second column shows the two-loop count.}
\end{table}

\subsection{Integral Families}
The integral families are assigned to the amplitude through few sequential steps.
Initially, a database is constructed containing all momentum routings that may arise in the forward-scattering amplitudes.
To accomplish this, we generate all graphs contributing to a generic process at a fixed loop order $L$, excluding tadpoles and allowing for only one type of internal edge to be present,
\begin{align}
    p + q \longrightarrow p + q \period
\end{align}
The momenta $p$ and $q$ denote any parton and boson, respectively.
It is evident that all graphs contributing to the physically relevant forward scattering amplitude are contained within this set of reference graphs, as the latter are determined solely by the forward kinematics.

Each internal edge $i$ is then assigned a momentum $\tilde{q}_{i}$ with an a priori arbitrary orientation.
At each vertex, momentum conservation implies a system of equations that is not of maximal rank.
These systems of equations are solved for vanishing $p$,  the momenta,
\begin{align}
\tilde{q}_i(q,l_1,\dots,l_L) \equiv \omega_{i0} q + \sum_{j=1}^L \omega_{i,j} l_j, \quad \omega_{i,j} \in \set{-1,0,1} \comma
\end{align}
where the $l_{j}$ depict the loop momenta.
The loop momenta emerge because the systems of equations are not of maximal rank.
From these solutions, it can be inferred which edges of a graph are bridges and which correspond to dotted propagators, because bridges are edges carrying momenta independent of the loop momenta, whilst dotted propagators arise when momenta, or rather their squares, occur multiple times.

In light of the harmonic projection, the four-point function reduces kinematically to a two-point function, as nullifying the partonic momentum $p$ is kinematically equivalent to amputating the partonic edges.
It should be noted that the amputation of the partonic edges may create additional dotted propagators.
To proceed, we define reduced graphs as amputated graphs for which all bridges and edges corresponding to dotted propagators are contracted until the graph is free of dots.
Notably, the resulting reduced graphs need not be unique, since no preferred choice exists for contracting dotted edges.
This is unimportant, as it merely represents a subtopology to which the four-point graph reduces under the harmonic projection (see, e.g., the discussion of two-loop subtopologies in Ref.~\cite{Moch:1999eb}).

To restore the partonic momenta and streamline the calculation, we minimise the number of edges carrying the partonic momentum.
A depth-first search algorithm achieves this by identifying one of the shortest possible paths through the graph, acknowledging that such a path need not be unique.
Once such a path is determined, we restore the $p$-dependence by adopting the convention that the incoming parton momentum flows into the graph, whilst the outgoing parton momentum flows out. At each vertex, momentum conservation then fixes the momentum routing unambiguously, provided a shortest path exists.
Each internal edge is then ascribed a momentum,
\begin{align}
p_i = \tilde{q}_i + \eta_i p,\quad \eta_i \in \set{-1,0,1} \period
\end{align}
In summary, the first step assigns each graph a topology, namely the graph itself, a reduced graph, and a shortest $p$-path.

To continue, we repeat this procedure for a generic two-point function involving only three-valent vertices,
\begin{align}
   q \longrightarrow q \period
\end{align}
Once all superfluous edges are contracted, the graphs may contain vertices of valence greater than three.
All graphs with valence greater than three are then discarded, as they can be derived from graphs with only three-valent vertices by contracting internal edges.
For the remaining graphs, an equivalence class is generated, consisting of all graphs obtainable by contracting their internal edges.
Since the contraction of edges respects momentum conservation, each equivalence class may be ascribed a unique set of momenta,
\begin{align}
q_i(q,l_1,\dots,l_L) \equiv \omega_{i0} q + \sum_{j=1}^L \omega_{i,j} l_j, \quad \omega_{i,j} \in \set{-1,0,1} \period
\end{align}
Then, reduced graphs of the four-point graphs are matched against the equivalence classes of the two-point graphs with edge orientation taken into account.
All momenta of the four-point graphs can then be expressed in terms of those of the two-point graphs.
Hence, the outlined procedure generates a reference database of four-point topologies with their momentum routings that allows the two-point momenta to be straightforwardly assigned to the graphs contributing to the forward scattering amplitude, as the routing depends only on the topology of the graphs and not on the edge types.

Lastly, the momenta of the two-point graphs must be expressed in terms of the loop momenta $l_i$ and the bosonic momentum $q$.
For this purpose, let us define a generic momentum routing,
\begin{align}
\mathcal{R}_L = \big\{ q_1, ..., q_n \big\}, \quad n = \frac{L(L+3)}{2} \comma
\end{align}
where $n$ is the number of irreducible scalar products that may be constructed from the loop momenta $l_i$ and the external momentum $q$.
It turns out that up to two loops, it is possible to map all two-point graphs onto a single momentum routing, 
\begin{subequations}
\begin{align}
\label{eq:R1}
    \mathcal{R}_1 &= \big\{ l_1,l_1+q \big\} \comma \LineNumber
\label{eq:R2}
    \mathcal{R}_2 &= \big\{ l_1,l_1+q,l_2,l_2+q,l_1-l_2 \big\} \comma
\end{align}
\end{subequations}
see \autoref{fig:mom-routing_graphs}. 
It is advantageous to use these momentum routings as they already minimise the number of integral families needed for the discussion of massless amplitudes.
Furthermore, they make cancellations already manifest at the integrand level, as the master integrals within a family are linearly independent, which is not necessarily the case across many families.
Therefore, the momentum routings would already suffice to define integral families if all internal lines in the graphs were massless.
However, since internal lines may be massive, the momentum routings must be dressed with masses.
\begin{figure}[ht]
  \centering
  \fbox{%
    \includegraphics[width=0.7\linewidth]{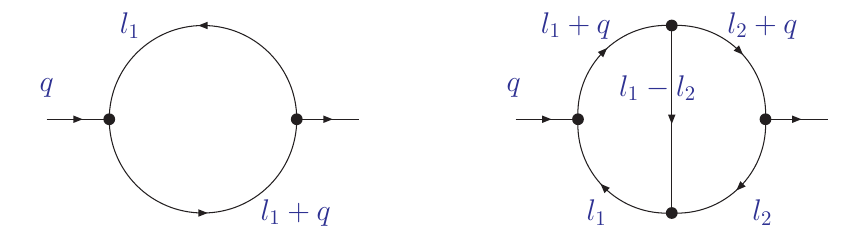}%
  }
  \caption{Independent two-point subtopologies and momentum routing at one and two loops, $\mathcal{R}_1$ and $\mathcal{R}_2$, see Eqs.~(\ref{eq:R1}) and (\ref{eq:R2}). 
  }
\label{fig:mom-routing_graphs}
\end{figure}

To incorporate the masses, we recall that edges were previously separated into massive and massless lines, with each edge $i$ assigned a mass $m_{i}$, equal either to zero or to the heavy-quark mass $m$.
With this convention, let $\mathcal{G}_{L}$ denote the set of all graphs contributing to the forward scattering amplitude with $L$ loops, and let $\mathcal{E}_{\Gamma}$ denote the set of edges of a graph $\Gamma$.
For each graph, we may define a set of pairs describing the structure of its denominators, and for each such set a set of propagators,
\begin{align}
\mathcal{F}_{\Gamma} &\equiv \bigcup_{i\in \mathcal{E}_\Gamma} \set{ (p_i\vert_{p=0},m_i)} \backslash \sim \period 
\end{align}
Here, equivalent momenta, that is, those differing by an overall sign, are identified as they appear always squared in the propagators.
Furthermore, we define an auxiliary set
\begin{align}\mathcal{R}(\mathcal{F}_\Gamma) &\equiv \bigcup_{ (p,m) \in \mathcal{F}_\Gamma} \{p\} \period
\end{align}
The set of denominator structures may further be promoted to an equivalence class by incorporating all sets of momenta that remain within the same momentum routing under momentum shift transformations.
For notational convenience, the momentum shifts are defined as transformations that are implemented through a linear operator $T$ satisfying,
\begin{subequations}
\begin{align}
    T\mathcal{F}_{\Gamma} &= \bigcup_{i\in \mathcal{E}_\Gamma} \set{ (T(p_i)\vert_{p=0},m_i)} \backslash \sim \comma \LineNumber
    \quad T(l_i) &= \omega_{i,0} q + \sum_{j=1}^L  \omega_{i,j} k_j, \quad \omega_{i,j} \in \set{-1,0,1} \period
\end{align}
\end{subequations}
The set of all momentum-shift transformations is denoted by $\mathcal{T}$.
Not all of these transformations are admissible, as only those preserving the set size and remaining within the same momentum routing are relevant.
The set of admissible transformations is given by
\begin{align}
    \mathcal{T}_A = \big\{ T \in \mathcal{T} \; \vert \; (\mathcal{R}(T\mathcal{F}_{\Gamma})  \backslash  R_L = \emptyset )\wedge (
    | T \mathcal{F}_{\Gamma} | = |\mathcal{F}_{\Gamma}| ) \big\} \period
\end{align}
Constraints on the size arise because some shifts may cause two or more momenta to coincide, making the momentum shift invalid as the corresponding Jacobian matrix would fail to be invertible.
An equivalence class of denominator structures may then be defined
\begin{align}
    \big[ \mathcal{F}_{\Gamma}  \big] = \bigcup_{T\in \mathcal{T}_A} \big\{ T \mathcal{F}_{ \Gamma} \big\} \period
\end{align}
These equivalence classes are free of duplicate entries as they are unions of sets. 
Remaining redundancies arise only from classes whose members are contained in subsets of the members of equivalence classes of larger cardinality.
This condition translates into the requirement that the power set of each equivalence class must not be a subset of the power set of any class of larger size.
We obtain the set of independent equivalence classes
\begin{align}
\mathcal{C}_L 
    = \bigcup_{\substack{
          \Gamma \in \mathcal{G}_L,\\
          P([\mathcal{F}_{\Gamma}])\not\subset P([\mathcal{F}_{\Gamma^\prime}])
          \ \forall\, \Gamma^\prime \in \mathcal{G}_L: |[\mathcal{F}_{\Gamma^\prime}]| > |[\mathcal{F}_{\Gamma}]|
        }} 
        \big\{ [ \mathcal{F}_{\Gamma}  ] \big\} \comma
\end{align}
where $P$ denotes the power set.
To find the optimal set of integral families, a given momentum routing is dressed with masses, giving rise to a basis such as
\begin{align}
    \mathcal{B}_{L,b} = \{(q_1,b_1),\dots,(q_n,b_n) \},\quad b= \sum_{i=1}^n 2^{i-1} \frac{ b_{n-i+1}}{m_i} \period
\end{align}
Constructing sets for all possible mass assignments of the momentum routings is equivalent to taking the union over all such sets for every possible representation of $b$ in binary with $n$ digits.
Therefore, the set of all integral families is obtained by taking the union,
\begin{align}
    \mathcal{B}_{L} = \bigcup_{b=0}^{2^n-1} \big\{ \mathcal{B}_{L,b} \big\} \period
\end{align}
We now introduce auxiliary functions for notational convenience.
First, a bijective index function $g$ that maps each independent equivalence class to an index in the set $\mathcal{I} \subset \mathbb{N}$ with $\abs{\mathcal{I}} =\abs{\mathcal{C}_\Gamma}$, where the index values may be chosen arbitrarily,
\begin{align}
    g_L&: \mathcal{C}_L \to \mathcal{I}, \quad g_L(C) = i, \quad C \in \mathcal{C}_L, i \in \mathcal{I} \period
\end{align}
Since the map is applied to equivalence classes, it induces a canonical map that assigns to each member of a class the index of the class itself,
\begin{align}
    \tilde{g}_L&: C\in\mathcal{C}_L \to \mathcal{I}, \quad \tilde{g}_L(c) = g_L(C), \quad c \in C \period
\end{align}
This follows because, by construction, the equivalence classes are disjoint, and hence their members must be independent.
Otherwise, the classes would not be well defined, as dependencies among their members would imply that the classes themselves are not independent.
Second, we use the index function to define an indicator function,
\begin{align}
    f_L&: \mathcal{C}_L \times \mathcal{B} \to P_{\leq1}(\mathcal{I}), \quad f_L(c, b) =
    \begin{cases}
    \{ \tilde{g}_L(c) \}, & c \setminus b = \emptyset, \\
    \emptyset,  & c \setminus b \neq \emptyset,
    \end{cases} \quad c \in \mathcal{C}_L, b \in \mathcal{B}\comma
\end{align}
where $P_{\leq 1}$ denotes the power set generating either the empty set or sets of cardinality one.
Using the indicator function, each basis family may be ascribed independent classes,
\begin{align}
\mathcal{S}_{L,b} = \bigcup_{C\in \mathcal C} \bigg( \bigcup_{c\in C} f_L(c,\mathcal{B}_{L,b}) \bigg) \period
\end{align}
From this family of sets, finding a minimal set of integral families reduces to a standard set-covering problem, with the family $\mathcal{S}_L$ and the universe $\mathcal{U}_L$ given by
\begin{subequations}
\begin{align}
    \mathcal{S}_L &= \bigcup_{b=0}^{2^n-1} \big\{ \mathcal{S}_{L,b} \big\} \comma \LineNumber
    \mathcal{U}_L &= \bigcup_{b=0}^{2^n-1} \mathcal{S}_{L,b} \period
\end{align}
\end{subequations}
Up to three loops, while here only shown for two loops, the problem remains sufficiently small to be treated with non-greedy implementations of set-covering algorithms.
This is feasible because the covering sets are sparse, with most entries being empty, which is only due to the introduction of independent equivalence classes.
Hence, the integral families found this way constitute a minimal set.
It is worth noting that the number of families can be greatly reduced by considering only families that are inequivalent under momentum shifts, rather than all families.
Solving the set-covering problem then yields the integral families,
\begin{subequations}
\begin{align}
    \hat{\mathcal{B}}_1 &= \big\{ \mathcal{B}_{1,0}, \mathcal{B}_{1,3} \big\} \comma \LineNumber
    \hat{\mathcal{B}}_2 &= \big\{ \mathcal{B}_{2,0}, \mathcal{B}_{2,30} \big\} \period
\end{align}
\end{subequations}
In a final step, when assigning the integral families to the graphs of the forward scattering amplitude, we must account for the fact that the assignment was performed using equivalence classes.
Hence, the momenta in the amplitude are shifted accordingly, determining these shifts is only a matter of book-keeping rather than a conceptual difficulty.

\subsection{Master Integrals}

The calculation of the fixed Mellin moments takes place in close analogy to that of the Mellin moments of massless structure functions within the \texttt{MINCER} framework (see Refs.~\cite{Gorishnii:1989gt,Larin:1991fz}), but is adjusted to allow for massive internal lines.
One difference, however, is that the resulting integrals are not directly reduced within \texttt{FORM}, as the symbolic reduction of Feynman integrals with massive propagators is more involved than in the massless case.
The reason is twofold.
First, the zero sector of the integral families is larger in the massless case than in the massive case, and more integration-by-parts identities are needed in the massive case.
Second, in the massless case propagators correspond to scalar products, which makes it straightforward to impose integral symmetries on the fly, whereas this is not the case in the presence of masses. 

When calculating the Mellin moments, the first step is the application of the Feynman rules and subsequently the harmonic projection to the previously defined super-graphs.
All remaining scalar products are then expressed in terms of the denominators of the assigned family.
The resulting denominators are collected in integrals of the form
\begin{align}
    B_{L,b}(a_1, \dots ,a_n) = e^{\epsilon L ( \log(4\pi) - \gamma_\text{E}) } \int \left( \prod_{j=1}^L \frac{\dd[D]{l_j}}{(2 \pi)^D} \right) \left(  \prod_{i=1}^n \frac{1}{[q_i^2 - b_i m^2]^{a_i}} \right) \period
\end{align}
By convention, all integrals are evaluated in the $\overline{\text{MS}}$-scheme.
These integrals are simplified using integration-by-parts relations, with the Laporta algorithm implemented in \texttt{KIRA} (see Ref.~\cite{Lange:2025fba}).
For the derived integral families, one finds that the following master integrals form a complete basis,
\begin{subequations}
\begin{align}
    \big\langle \mathcal{B}_{1,0} \big\rangle &\equiv \big\langle B_{1,0}(1,1) \big\rangle_{\mathbb{Q}(D,\sp{q}{q})}  \comma \LineNumber
    \big\langle \mathcal{B}_{1,3} \big\rangle  &\equiv \big\langle B_{1,3}(1,0), B_{1,3}(1,1) \big\rangle_{\mathbb{Q}(D,\sp{q}{q},m^2)} \comma
    \LineNumber
    \big\langle \mathcal{B}_{2,0} \big\rangle  &\equiv \big\langle B_{2,0}(1,1,1,1,0),B_{2,0}(0,1,1,0,1) \big\rangle_{\mathbb{Q}(D,\sp{q}{q})} \comma \LineNumber
    \big\langle \mathcal{B}_{2,30} \big\rangle  &\equiv \big\langle B_{2,30}(1,0,1,0,0),B_{2,30}(1,1,1,0,0),B_{2,30}(1,1,1,1,0),\LineNoNumber 
    &\qquad B_{2,30}(0,1,1,0,1),B_{2,30}(0,1,1,0,2) \big\rangle_{\mathbb{Q}(D,\sp{q}{q},m^2)} \period
\end{align}
\end{subequations}
Here, $\langle \dots \rangle_{\mathbb{Q}(x_1, \dots, x_n)}$ denotes the linear span over the field of rational functions in the variables $x_1, \dots, x_n$ with rational coefficients.

Next, we briefly discuss the master integrals required for this calculation. These integrals are well known in the literature and have been analysed extensively, for example in Ref.~\cite{Bauberger:1994hx,Martin:2005qm}, and more recently in Ref.~\cite{Forner:2024ojj}. The master integrals appearing in our setup are not all independent, as several two‑loop topologies can be reduced to products of one‑loop integrals. For completeness, and to make the presentation self‑contained, we outline their solution using the method of differential equations.
\begin{align}
B_{2,0}(1,1,1,1,0) &= B_{1,3}(1,1)^2 \comma \LineNumber
B_{2,30}(1,0,1,0,0) &= B_{1,3}(1,0)^2  \comma \LineNumber 
B_{2,30}(1,1,1,0,0) &= B_{1,3}(1,0) B_{1,3}(1,1) \comma  \LineNumber 
B_{2,30}(1,1,1,1,0) &= B_{1,3}(1,1)^2 \period
\end{align}
To evaluate the master integrals, the massless master integrals are straightforward, as they reduce to Euler beta functions and convolutions thereof. 
In particular, one finds
\begin{subequations}
\begin{align}
B_{1,0}(1,1) &= e^{\epsilon ( \log(4\pi) - \gamma_\text{E} )} (-\sp{q}{q})^{-\epsilon} \frac{\gammaf{1-\epsilon}^2 \gammaf{\epsilon} }{\gammaf{2(1-\epsilon)}} \comma \LineNumber 
B_{2,0}(0,1,1,0,1) &= e^{2 \epsilon ( \log(4\pi) - \gamma_\text{E} )} (-\sp{q}{q})^{1-3\epsilon} \frac{\gammaf{1-\epsilon}^3 \gammaf{2\epsilon-1} }{\gammaf{3(1-\epsilon)}}  \period
\end{align}
\end{subequations}
For the remaining massive integral families, the one-loop integrals need not be calculated explicitly, as they can be obtained from the two-loop results by taking square roots.
This determines the one-loop master integrals up to an overall sign, an ambiguity that will be fixed later.
To calculate these integrals, it is sufficient to set up a single system of differential equations by defining the vector,
\begin{align}
    \vec{B}  =
    \begin{bmatrix}
    B_{2,30}(1,0,1,0,0) \\ 
    B_{2,30}(1,1,1,0,0)  \\ 
    B_{2,30}(1,1,1,1,0) \\ 
    B_{2,30}(0,1,1,0,2)   \\ 
    B_{2,30}(0,1,1,0,1)
    \end{bmatrix} \period
\end{align}
All other master integrals can be deduced from the master integrals calculated in this way.
To simplify the calculation, scaling properties of the Feynman integrals are used, allowing one to work in a chart with
\begin{align}
  \sp{Q}{Q} = - \sp{q}{q} = 1 \period
\end{align}
While it is trivial to restore the $\sp{Q}{Q}$ dependence, working in this chart has the advantage that only a univariate differential equation needs to be considered rather than a multivariate one.
To simplify the notation, the following change of variables is employed,
\begin{align}
  \kappa \equiv \frac{m^2}{\sp{Q}{Q}} \period
\end{align}
In this chart, and using the new variable, the differential equation is obtained by using integration-by-parts identities
\begin{align}
    \partial_\kappa  \vec{B} \dd{\kappa} = A_{\kappa } \dd{\kappa} \comma
\end{align}
where the connection one-form is obtained by taking derivatives of the master integral vector with respect to $\kappa$ and applying integration-by-parts identities.
The resulting one-form, after partial fractioning, contains denominators quadratic in $\kappa$.
To bring the differential equations into canonical form, the square roots are rationalised.
This is accomplished by the change of variables,
\begin{subequations}
\begin{align}
	\kappa &= \frac{\lambda}{(1-\lambda)^2} \comma \LineNumber
       \dd{\kappa} &= - \frac{1+\lambda}{(1-\lambda)^3}  \dd{\lambda} \period
\end{align}
\end{subequations}
Upon changing variables, the differential equation can be brought into canonical form using \texttt{CANONICA} (see Refs.~\cite{Meyer:2017joq,Meyer:2018feh}).
Upon performing a partial fraction decomposition, all master integrals can be expressed in terms of iterated integrals, which give rise to multiple polylogarithms over the alphabet,
\begin{align}
    \mathcal{A} = \set{\lambda,\lambda-1,\lambda+1} \period
\end{align}
Therefore, a subclass of multiple polylogarithms, namely the harmonic polylogarithms (see Ref.~\cite{Remiddi:1999ew}), suffices.
These are defined by
\begin{subequations}
\begin{align}
\label{eq:hpl_definition}
    \hpl{a_0,a_1,\dots,a_n} &\equiv (-1)^{\delta_{a,1}}\integral{\lambda^\prime}{0}{\lambda} \frac{\hpl{a_1,\dots,a_n}}{\lambda^\prime-a}\comma \; (a_i \in \set{-1,0,1}) \wedge \neg\!\left( \forall i, a_i = 0 \right)\comma \LineNumber
     \hpl{a_0,a_1,\dots,a_n} &\equiv \frac{1}{n!} \log^n (\lambda), \quad a_i \in \set{0} \period
\end{align}
\end{subequations}
To solve the integral, it is assumed that the Feynman integrals admit a series expansion,
\begin{align}
    \vec{B} = \sum_{k\in \mathbb{Z}} \vec{B}_k \epsilon^k \period
\end{align}
Given boundary values, the differential equation may be solved order by order as a series in the regulator $\epsilon$.
All boundary conditions can be obtained from the massless integrals and from massive integrals evaluated at specific values of $\lambda$, or equivalently $\kappa$.
For this purpose, the integrals were calculated using the Mellin–Barnes method and evaluated in terms of hypergeometric functions.
The boundary conditions for all but the last integral are then obtained by evaluating the exact result for specific values of $\kappa$.
This simultaneously fixes the signs of the one-loop integrals, eliminating the ambiguity that arose from extracting square roots of the two-loop integrals.

All results were numerically verified at various kinematic points using \texttt{AMFlow} (see Ref.~\cite{Liu:2022chg}) and compared with the exact results calculated for the boundary conditions.
For the calculation of the fixed Mellin moments at two loops, the integrals need only be expanded to $\order{\epsilon^3}$, as higher orders are not required for the two-loop coefficient functions.

\subsection{Mellin moments at large values of $N$}
\label{sec:Mellin-large}

In an independent approach we generate the amplitude for the forward scattering, without expansing in the momentum $p$. 
In this approach we generate the contributing graphs with \texttt{QGRAF}~\cite{Nogueira:1991ex} and then use the packages \texttt{tapir}~\cite{Gerlach:2022qnc} and \texttt{exp}~\cite{Harlander:1998cmq,Seidensticker:1999bb} to translate them into \texttt{FORM}~\cite{Vermaseren:2000nd,Kuipers:2012rf,Ruijl:2017dtg,FORM_v4.3.1} input and map them to integral families. 
Subsequently, we use \texttt{FORM} and the \texttt{color.h}~\cite{vanRitbergen:1998pn} package to perform Dirac- and colour-algebra and express the amplitude in terms of scalar integrals. 
Due to the forward kinematics of our amplitude we need to apply partial fraction decomposition and we end up with a large number of redundant integral families. 
We use the package \texttt{feynson}~\cite{Maheria:2022dsq} to reduce them to a minimal set. 
It turns out that we need 2 (8) integrals families at one (two) loop order. 
Using symmetry properties under $q \leftrightarrow -q$, we can reduce the number of families to 1 (4).
In the limit $1/z = \omega \to 0$ the families can be mapped to the integral families which have been introduced for the calculation of fixed moments before. 

Subsequently, we use \texttt{KIRA}~\cite{Lange:2025fba} to reduce the amplitude to master integrals. 
Since our reductions depend on the two dimensionless scales $\kappa = m^2/Q^2$ and $\omega$ as well as the dimension $D=4-2\epsilon$, it turns out that it is beneficial to first search for a basis of master integrals where the denominators appearing in the reduction tables already factorise in the kinematic variables $\kappa$, $\omega$ and the dimension $D$.
We use the package \texttt{FindBetterBasis}~\cite{Smirnov:2020quc} which is shipped together with \texttt{FIRE6}~\cite{Smirnov:2023yhb} to find these candidates. 
In the end, we find a total of 72 master integrals. 
The reduction to master integrals also allows us to derive differential equations in $\omega$ and $\kappa$. 

The differential equation in $\omega$ can now be exploited to calculate high moments of the structure functions.
In practise, we insert an ansatz for the expansion around $\omega \to 0$, a simple Taylor series, into the differential equation.
This turns the differential equation into a linear system of equations for the expansion coefficients. 
We solve this system with \texttt{KIRA} together with \texttt{FireFly}~\cite{Klappert:2019emp,Klappert:2020aqs}. 
The boundary conditions of the differential equations are given by the master integrals evaluated at $\omega=0$.
These can be reduced to the integrals discussed and solved in the previous section and we take the results directly from there.

This approach is effective for the number of moments we want to achieve in the massive case.
We note that there are more efficient algorithms based on recurrences, which allow for the calculation of much higher moments, see, e.g. Refs.~\cite{Blumlein:2017dxp,Basdew-Sharma:2022vya}. 
However, these methods are particularly useful for one-scale problems, when analytical results can be reconstructed from a large number of moments. 
In our case, the rational functions in $\kappa$ grow for increasing moments, making a reconstruction of the analytic result from moments rather unfeasible.

We calculate the master integrals up to and including $\mathcal{O}(\omega^{25})$.
Due to inverse powers of $\omega$ in the amplitude this allows us to calculate up to the 23rd Mellin moment of the structure functions.
In the first approach based on expansions and harmonic projectors,
we are faced with a factorial growth of the amplitude and complexity of the scalar integrals to be reduced.
Going to $N=22$ requires, for example, the reduction of integrals with tensor rank and dots up to 24. 
To perform such reductions efficiently one usually has to resort to custom codes such like \texttt{MINCER}~\cite{Gorishnii:1989gt,Larin:1991fz}, \texttt{MATAD}~\cite{Steinhauser:2000ry} or \texttt{FORCER}~\cite{Ruijl:2017cxj}, which are, however, not available for the topologies we are dealing with.

\subsection{Going beyond fixed moments}
To obtain more complete results for the structure functions beyond fixed moments, two directions are possible. On the one hand, one can perform a complete analytic calculation of the structure functions. On the other hand, one can try and obtain systematic expansions in different kinematic limits. 

The analytic solutions of the NLO coefficient functions involving massive quarks are rather involved 
and have only been achieved in the quark initiated coefficient functions, see for example Ref.~\cite{Blumlein:2016xcy} for the non-singlet and Refs.~\cite{Blumlein:2019qze,Blumlein:2019zux} for the pure-singlet case. 
In the latter case, iterated integrals over involved square-root valued letters are needed, which are related to incomplete elliptic integrals and are not ideally suited for fast numerical evaluations.
A complete analytic solution of the structure functions therefore also has to address this issue.

One option to improve the precision, while maintaining a simpler analytic form, is the expansion for $\kappa \ll 1$.
The leading term in this expansion is known from the calculation of massive operator matrix elements.
We have reproduced these results by solving the differential equations of the master integrals in this asymptotic limit.
Calculating additional power-suppressed expansion terms in the asymptotic limit is possible, however not easily applicable to NNLO.
We leave studies beyond fixed Mellin moments for future work.

\section{Results}
\label{sec:results}

\subsection{Analytical results for fixed Mellin moments}
The complete results for the Mellin moments of the transverse and longitudinal heavy quark coefficient functions, $C^N_{2,i}$ and 
 $C^N_{L,i}$ up to NLO, retaining their full mass dependence, are rather lengthy.
The set‑up discussed in the previous section also allows us to compute the Mellin moments of the massless DIS coefficient functions up to two loops, which are included in our results. The two‑loop massless contributions to $F_2$ and $F_L$ were first computed in Refs.~\cite{Zijlstra:1992qd, Moch:1999eb}, and our results agree, of course.
The full set of Mellin moments for $N=2,\ldots,22$ is provided in an ancillary file. 
The file can also be obtained from the authors upon request.

For illustration, we only explicitly present the results for the first moments.
Here and below, we fix the scales for the coefficient functions to $\mu_r = \mu_f = Q$.
To enhance readability, the expressions are decomposed, and the
polynomials $\num{i}{j}{k}$ appearing in the numerators are given separately.
In the following, $\hpl{\vec{a}}$ refers to the harmonic polylogarithms defined in~\eqref{eq:hpl_definition}, and $\zeta_n $ denotes the Riemann zeta function for an integer $n>1$. For $N =2 $ the moments are
\begin{align}
C^2_{2,\psg} &= \phantom{+} a_{s,\ren} \bigg[T_F {\langle\heavy_{\set{\vector,\vector}} \rangle^\psg} n_{\heavy} \bigg(\frac{4 \hpl{0} \num{\psg}{2}{02}}{3 (\lambda -1) (\lambda +1)^3}+\frac{\num{\psg}{2}{01}}{(\lambda +1)^2}\bigg)-T_F {\langle\light_{\set{\vector,\vector}} \rangle^\psg} n_{\light}\bigg] 
  \LineNoNumber &\quad + a_{s,\ren}^2 \bigg[C_A T_F {\langle\heavy_{\set{\vector,\vector}} \rangle^\psg} n_{\heavy} \bigg(\frac{4 \hpl{-1,0} \num{\psg}{2}{13}}{45 (\lambda -1) \lambda ^3 (\lambda +1)^3}-\frac{2 \hpl{1,0} \num{\psg}{2}{14}}{45 (\lambda -1) \lambda ^3 (\lambda +1)^3}
  \LineNoNumber &\qquad +\frac{\hpl{0} \num{\psg}{2}{11}}{135 (\lambda -1) \lambda ^2 (\lambda +1)^3}-\frac{88 \hpl{0,1} \num{\psg}{2}{02}}{9 (\lambda -1) (\lambda +1)^3}+\frac{22 \hpl{1} \num{\psg}{2}{08}}{3 (\lambda +1)^2}+\frac{16 \hpl{-1,0,0} \num{\psg}{2}{09}}{15 (\lambda -1)^2}
  \LineNoNumber &\qquad -\frac{16 \hpl{0,-1,0} \num{\psg}{2}{09}}{15 (\lambda -1)^2}+\frac{8 \hpl{0,1,0} \num{\psg}{2}{09}}{15 (\lambda -1)^2}-\frac{8 \hpl{1,0,0} \num{\psg}{2}{09}}{15 (\lambda -1)^2}-\frac{2 \hpl{0,0} \num{\psg}{2}{12}}{45 (\lambda -1)^2 (\lambda +1)^4}
  \LineNoNumber &\qquad +\frac{\num{\psg}{2}{10}}{90 \lambda ^2 (\lambda +1)^2}-\frac{4 \mzv{3} \num{\psg}{2}{09}}{5 (\lambda -1)^2}\bigg)+C_F T_F {\langle\heavy_{\set{\vector,\vector}} \rangle^\psg} n_{\heavy} \bigg(\frac{4 \hpl{-1,0} \num{\psg}{2}{06}}{45 (\lambda -1) \lambda ^3 (\lambda +1)}
  \LineNoNumber &\qquad -\frac{2 \hpl{1,0} \num{\psg}{2}{06}}{45 (\lambda -1) \lambda ^3 (\lambda +1)}+\frac{\hpl{0} \num{\psg}{2}{05}}{45 (\lambda -1) \lambda ^2 (\lambda +1)^3}-\frac{64 \hpl{-1,0,0} \num{\psg}{2}{03}}{15 (\lambda -1)^2}+\frac{64 \hpl{0,-1,0} \num{\psg}{2}{03}}{15 (\lambda -1)^2}
  \LineNoNumber &\qquad -\frac{32 \hpl{0,1,0} \num{\psg}{2}{03}}{15 (\lambda -1)^2}+\frac{32 \hpl{1,0,0} \num{\psg}{2}{03}}{15 (\lambda -1)^2}-\frac{2 \hpl{0,0} \num{\psg}{2}{07}}{45 (\lambda -1)^2 (\lambda +1)^4}+\frac{\num{\psg}{2}{04}}{90 \lambda ^2 (\lambda +1)^2}
  \LineNoNumber &\qquad  +\frac{16 \mzv{3} \num{\psg}{2}{03}}{5 (\lambda -1)^2}\bigg)+\bigg(\frac{115}{162}-4 \mzv{3}\bigg) C_A T_F {\langle\light_{\set{\vector,\vector}} \rangle^\psg} n_{\light}+\bigg(\frac{32 \mzv{3}}{5}-\frac{4799}{405}\bigg) C_F T_F {\langle\light_{\set{\vector,\vector}} \rangle^\psg} n_{\light}\bigg]
\end{align}
\begin{subequations}
\begin{align}
\num{\psg}{2}{01} &= -\lambda ^2+4 \lambda -1\phantom{\big(} \LineNumber
\num{\psg}{2}{02} &= \lambda ^4-3 \lambda ^3+\lambda ^2-3 \lambda +1\phantom{\big(} \LineNumber
\num{\psg}{2}{03} &= 2 \lambda ^2-3 \lambda +2\phantom{\big(} \LineNumber
\num{\psg}{2}{04} &= -124 \lambda ^6-683 \lambda ^5-252 \lambda ^4+2318 \lambda ^3-252 \lambda ^2-683 \lambda -124\phantom{\big(} \LineNumber
\num{\psg}{2}{05} &= -62 \lambda ^8-295 \lambda ^7+262 \lambda ^6-2373 \lambda ^5+4152 \lambda ^4-2373 \lambda ^3+262 \lambda ^2-295 \lambda \LineNoNumber
  &\quad -62\phantom{\big(} \LineNumber
\num{\psg}{2}{06} &= 31 \lambda ^8+132 \lambda ^7-12 \lambda ^6-456 \lambda ^5+514 \lambda ^4-456 \lambda ^3-12 \lambda ^2+132 \lambda +31\phantom{\big(} \LineNumber
\num{\psg}{2}{07} &= 93 \lambda ^9+582 \lambda ^8+756 \lambda ^7-1546 \lambda ^6-3980 \lambda ^5+4650 \lambda ^4-4312 \lambda ^3+3258 \lambda ^2 \LineNoNumber
  &\quad -1901 \lambda +80\phantom{\big(} \LineNumber
\num{\psg}{2}{08} &= \lambda ^2-4 \lambda +1\phantom{\big(} \LineNumber
\num{\psg}{2}{09} &= 5 \lambda ^2-8 \lambda +5\phantom{\big(} \LineNumber
\num{\psg}{2}{10} &= -116 \lambda ^6-585 \lambda ^5-275 \lambda ^4-338 \lambda ^3-275 \lambda ^2-585 \lambda -116\phantom{\big(} \LineNumber
\num{\psg}{2}{11} &= -174 \lambda ^8-747 \lambda ^7+1828 \lambda ^6-3045 \lambda ^5+4192 \lambda ^4+915 \lambda ^3+838 \lambda ^2-747 \lambda  \LineNoNumber
  &\quad -174\phantom{\big(} \LineNumber
\num{\psg}{2}{12} &= 87 \lambda ^9+504 \lambda ^8+396 \lambda ^7-516 \lambda ^6-474 \lambda ^5-564 \lambda ^4+892 \lambda ^3+30 \lambda ^2+315 \lambda  \LineNoNumber
  &\quad -110\phantom{\big(} \LineNumber
\num{\psg}{2}{13} &= 29 \lambda ^{10}+168 \lambda ^9+161 \lambda ^8-114 \lambda ^7+338 \lambda ^6-312 \lambda ^5+338 \lambda ^4-114 \lambda ^3 \LineNoNumber
  &\quad +161 \lambda ^2+168 \lambda +29\phantom{\big(} \LineNumber
\num{\psg}{2}{14} &= 29 \lambda ^{10}+168 \lambda ^9+161 \lambda ^8+106 \lambda ^7-322 \lambda ^6-92 \lambda ^5-322 \lambda ^4+106 \lambda ^3+161 \lambda ^2 \LineNoNumber
  &\quad +168 \lambda +29\phantom{\big(}
\end{align}
\end{subequations}
\begin{align}
C^2_{2,\psq} &= 1+\frac{a_{s,\ren} 
\langle\light_{\set{\vector,\vector}} \rangle^\psq C_F}{3}
  + a_{s,\ren}^2 \bigg[C_F T_F {\langle\heavy_{\set{\vector,\vector}} \rangle^\psg} n_{\heavy} \bigg(-\frac{8 \hpl{1,0} \num{\psq}{2}{06}}{45 \lambda ^3}+\frac{16 \hpl{-1,0} \num{\psq}{2}{11}}{45 (\lambda -1) \lambda ^3 (\lambda +1)^3} \LineNoNumber 
  &\quad -\frac{4 \hpl{0} \num{\psq}{2}{10}}{135 (\lambda -1) \lambda ^2 (\lambda +1)^3}-\frac{16 \hpl{1} \num{\psq}{2}{05}}{3 (\lambda +1)^2}+\frac{64 \hpl{0,1} \num{\psq}{2}{07}}{9 (\lambda -1) (\lambda +1)^3}-\frac{8 \hpl{0,0} \num{\psq}{2}{09}}{45 (\lambda -1)^2 (\lambda +1)^3}\LineNoNumber 
  &\quad -\frac{2 \num{\psq}{2}{08}}{45 \lambda ^2 (\lambda +1)^2}\bigg)+\langle\heavy_{\set{\vector,\vector}} \rangle^\psg C_F T_F n_{\heavy} \bigg(\frac{2 \hpl{1} \num{\psq}{2}{01}}{14175 \lambda ^3}-\frac{8 \hpl{0,0} \num{\psq}{2}{02}}{45 (\lambda -1)^4}+\frac{2 \hpl{0} \num{\psq}{2}{04}}{14175 (\lambda -1)^3}\LineNoNumber 
  &\quad +\frac{\num{\psq}{2}{03}}{42525 (\lambda -1)^2 \lambda ^2}\bigg)+\langle\light_{\set{\vector,\vector}} \rangle^\psq \bigg( \bigg(\frac{3677}{135}-\frac{128 \mzv{3}}{5}\bigg) C_A C_F-8 C_F T_F n_{\light}\LineNoNumber 
  &\quad +\bigg(\frac{96 \mzv{3}}{5}-\frac{4189}{810}\bigg) C_F{}^2\bigg) \bigg) -\frac{266}{81} C_F T_F {\langle\light_{\set{\vector,\vector}} \rangle^\psg} n_{\light}
\end{align}
\begin{subequations}
\begin{align}
\num{\psq}{2}{01} &= (\lambda -1)^2 \left(8 \lambda ^4+103 \lambda ^3-5262 \lambda ^2+103 \lambda +8\right)\phantom{\big(} \LineNumber
\num{\psq}{2}{02} &= 20 \lambda ^4-80 \lambda ^3+141 \lambda ^2-80 \lambda +20\phantom{\big(} \LineNumber
\num{\psq}{2}{03} &= -48 \lambda ^6-450 \lambda ^5-508273 \lambda ^4+1114802 \lambda ^3-508273 \lambda ^2-450 \lambda -48\phantom{\big(} \LineNumber
\num{\psq}{2}{04} &= 8 \lambda ^6+63 \lambda ^5-5697 \lambda ^4+64113 \lambda ^3-76167 \lambda ^2-22050 \lambda +36750\phantom{\big(} \LineNumber
\num{\psq}{2}{05} &= \lambda ^2-4 \lambda +1\phantom{\big(} \LineNumber
\num{\psq}{2}{06} &= (\lambda -1) (\lambda +1)^3 \left(7 \lambda ^2-4 \lambda +7\right)\phantom{\big(} \LineNumber
\num{\psq}{2}{07} &= \lambda ^4-3 \lambda ^3+\lambda ^2-3 \lambda +1\phantom{\big(} \LineNumber
\num{\psq}{2}{08} &= 28 \lambda ^6+75 \lambda ^5+199 \lambda ^4-1016 \lambda ^3+199 \lambda ^2+75 \lambda +28\phantom{\big(} \LineNumber
\num{\psq}{2}{09} &= 21 \lambda ^8+51 \lambda ^7-15 \lambda ^6-45 \lambda ^5-270 \lambda ^4+258 \lambda ^3+68 \lambda ^2-80 \lambda +20\phantom{\big(} \LineNumber
\num{\psq}{2}{10} &= 42 \lambda ^8+81 \lambda ^7-8 \lambda ^6-408 \lambda ^5+538 \lambda ^4+312 \lambda ^3-188 \lambda ^2+81 \lambda +42\phantom{\big(} \LineNumber
\num{\psq}{2}{11} &= 7 \lambda ^{10}+24 \lambda ^9+19 \lambda ^8+24 \lambda ^7-146 \lambda ^6+24 \lambda ^5-146 \lambda ^4+24 \lambda ^3+19 \lambda ^2+24 \lambda \LineNoNumber
&\quad +7\phantom{\big(}
\end{align}
\end{subequations}
\begin{align}
  C^2_{L,\psg} &= a_{s,\ren} \bigg[ T_F {\langle\heavy_{\set{\vector,\vector}} \rangle^\psg} n_{\heavy} \bigg(\frac{4 \num{\psg}{L}{01}}{3 (\lambda +1)^4}-\frac{8 \hpl{0} \num{\psg}{L}{02}}{(\lambda -1) (\lambda +1)^5}\bigg)+\frac{4}{3} T_F {\langle\light_{\set{\vector,\vector}} \rangle^\psg} n_{\light}\bigg] \LineNoNumber
  &\quad + a_{s,\ren}^2 \bigg[ C_A T_F {\langle\heavy_{\set{\vector,\vector}} \rangle^\psg} n_{\heavy} \bigg(\frac{4 \hpl{1,0} \num{\psg}{L}{04}}{15 (\lambda -1) \lambda ^3 (\lambda +1)}-\frac{8 \hpl{-1,0} \num{\psg}{L}{08}}{15 (\lambda -1) \lambda ^3 (\lambda +1)^5} \LineNoNumber
  &\quad +\frac{2 \hpl{0} \num{\psg}{L}{07}}{45 (\lambda -1) \lambda ^2 (\lambda +1)^5}-\frac{88 \hpl{1} \num{\psg}{L}{01}}{9 (\lambda +1)^4}+\frac{176 \hpl{0,1} \num{\psg}{L}{02}}{3 (\lambda -1) (\lambda +1)^5}-\frac{64 \hpl{-1,0,0} \num{\psg}{L}{03}}{5 (\lambda -1)^2} \LineNoNumber
  &\quad +\frac{64 \hpl{0,-1,0} \num{\psg}{L}{03}}{5 (\lambda -1)^2}-\frac{32 \hpl{0,1,0} \num{\psg}{L}{03}}{5 (\lambda -1)^2}+\frac{32 \hpl{1,0,0} \num{\psg}{L}{03}}{5 (\lambda -1)^2}+\frac{4 \hpl{0,0} \num{\psg}{L}{06}}{15 (\lambda -1)^2 (\lambda +1)^6} \LineNoNumber
  &\quad +\frac{\num{\psg}{L}{05}}{135 \lambda ^2 (\lambda +1)^4}+\frac{48 \mzv{3} \num{\psg}{L}{03}}{5 (\lambda -1)^2}\bigg)+C_F T_F {\langle\heavy_{\set{\vector,\vector}} \rangle^\psg} n_{\heavy} \bigg(-\frac{8 \hpl{-1,0} \num{\psg}{L}{11}}{15 (\lambda -1) \lambda ^3 (\lambda +1)^3}  \LineNoNumber
  &\quad +\frac{4 \hpl{1,0} \num{\psg}{L}{11}}{15 (\lambda -1) \lambda ^3 (\lambda +1)^3}+\frac{2 \hpl{0} \num{\psg}{L}{12}}{15 (\lambda -1) \lambda ^2 (\lambda +1)^5}+\frac{128 \hpl{-1,0,0} \num{\psg}{L}{09}}{15 (\lambda -1)^2}-\frac{128 \hpl{0,-1,0} \num{\psg}{L}{09}}{15 (\lambda -1)^2} \LineNoNumber
  &\quad+\frac{64 \hpl{0,1,0} \num{\psg}{L}{09}}{15 (\lambda -1)^2}-\frac{64 \hpl{1,0,0} \num{\psg}{L}{09}}{15 (\lambda -1)^2}+\frac{4 \hpl{0,0} \num{\psg}{L}{13}}{15 (\lambda -1)^2 (\lambda +1)^6}+\frac{\num{\psg}{L}{10}}{15 \lambda ^2 (\lambda +1)^4}-\frac{32 \mzv{3} \num{\psg}{L}{09}}{5 (\lambda -1)^2}\bigg) \LineNoNumber
  &\quad +\frac{346}{27} C_A T_F {\langle\light_{\set{\vector,\vector}} \rangle^\psg} n_{\light}+\bigg(-\frac{32 \mzv{3}}{5}-\frac{232}{135}\bigg) C_F T_F {\langle\light_{\set{\vector,\vector}} \rangle^\psg} n_{\light}\bigg]
\end{align}
\begin{subequations}
\begin{align}
\num{\psg}{L}{01} &= \big(\lambda ^2+\lambda +1\big) \big(\lambda ^2+8 \lambda +1\big)\phantom{\big(} \LineNumber
\num{\psg}{L}{02} &= \lambda  \left(\lambda ^4+2 \lambda ^3+4 \lambda ^2+2 \lambda +1\right)\phantom{\big(} \LineNumber
\num{\psg}{L}{03} &= \lambda\phantom{\big(} \LineNumber
\num{\psg}{L}{04} &= 29 \lambda ^8-150 \lambda ^7+152 \lambda ^6-58 \lambda ^5+102 \lambda ^4-58 \lambda ^3+152 \lambda ^2-150 \lambda +29\phantom{\big(} \LineNumber
\num{\psg}{L}{05} &= 1044 \lambda ^8-2007 \lambda ^7-6056 \lambda ^6\LineNoNumber
  &\quad +2037 \lambda ^5+12304 \lambda ^4+2037 \lambda ^3-6056 \lambda ^2-2007 \lambda +1044\phantom{\big(} \LineNumber
\num{\psg}{L}{06} &= \lambda  \big(87 \lambda ^{10}-102 \lambda ^9-909 \lambda ^8-600 \lambda ^7+1170 \lambda ^6+1532 \lambda ^5+130 \lambda ^4+67 \lambda ^2\LineNoNumber
  &\quad-110 \lambda -65\big)\phantom{\big(} \LineNumber
\num{\psg}{L}{07} &= 174 \lambda ^{10}-465 \lambda ^9-1500 \lambda ^8-1924 \lambda ^7-1810 \lambda ^6-5454 \lambda ^5+170 \lambda ^4+56 \lambda ^3\LineNoNumber
  &\quad-1280 \lambda ^2-465 \lambda +174\phantom{\big(} \LineNumber
\num{\psg}{L}{08} &= 29 \lambda ^{12}-34 \lambda ^{11}-274 \lambda ^{10}-234 \lambda ^9-9 \lambda ^8+20 \lambda ^7-428 \lambda ^6+20 \lambda ^5-9 \lambda ^4\LineNoNumber
  &\quad-234 \lambda ^3-274 \lambda ^2-34 \lambda +29\phantom{\big(} \LineNumber
\num{\psg}{L}{09} &= \lambda ^2+\lambda +1\phantom{\big(} \LineNumber
\num{\psg}{L}{10} &= 124 \lambda ^8-309 \lambda ^7-1238 \lambda ^6+189 \lambda ^5+1668 \lambda ^4+189 \lambda ^3-1238 \lambda ^2-309 \lambda \LineNoNumber
  &\quad +124\phantom{\big(} \LineNumber
\num{\psg}{L}{11} &= 31 \lambda ^{10}-116 \lambda ^9-177 \lambda ^8+152 \lambda ^7+90 \lambda ^6-344 \lambda ^5+90 \lambda ^4+152 \lambda ^3-177 \lambda ^2\LineNoNumber
  &\quad-116 \lambda +31\phantom{\big(} \LineNumber
\num{\psg}{L}{12} &= 62 \lambda ^{10}-201 \lambda ^9-560 \lambda ^8-16 \lambda ^7+1922 \lambda ^6+722 \lambda ^5+1922 \lambda ^4-16 \lambda ^3-560 \lambda ^2\LineNoNumber
  &\quad-201 \lambda +62\phantom{\big(} \LineNumber
\num{\psg}{L}{13} &= \lambda  \big(93 \lambda ^{10}-162 \lambda ^9-1227 \lambda ^8-792 \lambda ^7+2044 \lambda ^6+504 \lambda ^5-2872 \lambda ^4-6016 \lambda ^3\LineNoNumber
  &\quad -697 \lambda ^2-414 \lambda +259\big)\phantom{\big(}
\end{align}
\end{subequations}
\begin{align}
  C^2_{L,\psq} &=  a_{s,\ren} \frac{4 C_F}{3} \langle\light_{\set{\vector,\vector}} \rangle^\psq + a_{s,\ren}^2 \bigg[C_F T_F {\langle\heavy_{\set{\vector,\vector}} \rangle^\psg} n_{\heavy} \bigg(\frac{16 \hpl{1,0} \num{\psq}{L}{07}}{15 \lambda ^3 (\lambda +1)}-\frac{32 \hpl{-1,0} \num{\psq}{L}{11}}{15 (\lambda -1) \lambda ^3 (\lambda +1)^5} \LineNoNumber
  &\quad +\frac{8 \hpl{0} \num{\psq}{L}{10}}{45 (\lambda -1) \lambda ^2 (\lambda +1)^5}+\frac{64 \hpl{1} \num{\psq}{L}{05}}{9 (\lambda +1)^4}-\frac{128 \hpl{0,1} \num{\psq}{L}{06}}{3 (\lambda -1) (\lambda +1)^5}+\frac{16 \hpl{0,0} \num{\psq}{L}{09}}{15 (\lambda -1)^2 (\lambda +1)^5} \LineNoNumber
  &\quad +\frac{4 \num{\psq}{L}{08}}{135 \lambda ^2 (\lambda +1)^4}\bigg)+ \langle\heavy_{\set{\vector,\vector}} \rangle^\psq C_F T_F n_{\heavy} \bigg(-\frac{32 \hpl{1} \num{\psg}{L}{01}}{4725 \lambda ^3}+\frac{192 \hpl{0,0} \num{\psq}{L}{02}}{5 (\lambda -1)^4}-\frac{16 \hpl{0} \num{\psq}{L}{03}}{4725 (\lambda -1)^3} \LineNoNumber
  &\quad +\frac{8 \num{\psq}{L}{04}}{14175 (\lambda -1)^2 \lambda ^2}\bigg)+\langle\light_{\set{\vector,\vector}} \rangle^\psq \bigg(\bigg(\frac{2878}{135}-\frac{32 \mzv{3}}{5}\bigg) C_A C_F-\frac{184}{27} C_F T_F n_{\light} \LineNoNumber
  &\quad +\bigg(\frac{64 \mzv{3}}{5}-\frac{1906}{135}\bigg) C_F{}^2 \bigg) -\frac{160}{27} C_F T_F {\langle\light_{\set{\vector,\vector}} \rangle^\psg} n_{\light} \bigg]
\end{align}
\begin{subequations}
\begin{align}
\num{\psq}{L}{01} &= (\lambda -1)^6\phantom{\big(} \LineNumber
\num{\psq}{L}{02} &= \lambda ^2\phantom{\big(} \LineNumber
\num{\psq}{L}{03} &= 2 \lambda ^6-18 \lambda ^5+72 \lambda ^4-693 \lambda ^3-7623 \lambda ^2-7875 \lambda -525\phantom{\big(} \LineNumber
\num{\psq}{L}{04} &= 12 \lambda ^6-90 \lambda ^5-11783 \lambda ^4-110258 \lambda ^3-11783 \lambda ^2-90 \lambda +12\phantom{\big(} \LineNumber
\num{\psq}{L}{05} &= \left(\lambda ^2+\lambda +1\right) \left(\lambda ^2+8 \lambda +1\right)\phantom{\big(} \LineNumber
\num{\psq}{L}{06} &= \lambda  \left(\lambda ^4+2 \lambda ^3+4 \lambda ^2+2 \lambda +1\right)\phantom{\big(} \LineNumber
\num{\psq}{L}{07} &= (\lambda -1) \left(7 \lambda ^6-16 \lambda ^5-7 \lambda ^4-8 \lambda ^3-7 \lambda ^2-16 \lambda +7\right)\phantom{\big(} \LineNumber
\num{\psq}{L}{08} &= 252 \lambda ^8-261 \lambda ^7-1522 \lambda ^6+1653 \lambda ^5+3104 \lambda ^4+1653 \lambda ^3-1522 \lambda ^2-261 \lambda \LineNoNumber
  &\quad +252\phantom{\big(} \LineNumber
\num{\psq}{L}{09} &= \lambda  \big(21 \lambda ^9-27 \lambda ^8-132 \lambda ^7+36 \lambda ^6+285 \lambda ^5+33 \lambda ^4-126 \lambda ^3-86 \lambda ^2+8 \lambda +20\big) \LineNumber
\num{\psq}{L}{10} &= 42 \lambda ^{10}-75 \lambda ^9-204 \lambda ^8+397 \lambda ^7+604 \lambda ^6+420 \lambda ^5+244 \lambda ^4+37 \lambda ^3-244 \lambda ^2 \LineNoNumber
  &\quad -75 \lambda +42\phantom{\big(} \LineNumber
\num{\psq}{L}{11} &= 7 \lambda ^{12}-2 \lambda ^{11}-46 \lambda ^{10}-34 \lambda ^9+81 \lambda ^8+116 \lambda ^7+156 \lambda ^6+116 \lambda ^5+81 \lambda ^4\LineNoNumber
  &\quad-34 \lambda ^3  -46 \lambda ^2-2 \lambda +7\phantom{\big(}
\end{align}
\end{subequations}
\begin{align}
  C^{\ns,2}_{2,\psq} &= 1 +\frac{a_{s,\ren} C_F}{3}+a_{s,\ren}^2 \bigg[C_F T_F n_{\heavy} \bigg(\frac{2 \hpl{1} \num[\ns]{\psq}{2}{01}}{14175 \lambda ^3}-\frac{8 \hpl{0,0} \num[\ns]{\psq}{2}{02}}{45 (\lambda -1)^4}+\frac{2 \hpl{0} \num[\ns]{\psq}{2}{04}}{14175 (\lambda -1)^3} \LineNoNumber
  &\quad +\frac{\num[\ns]{\psq}{2}{03}}{42525 (\lambda -1)^2 \lambda ^2}\bigg)+\bigg(\frac{3677}{135}-\frac{128 \mzv{3}}{5}\bigg) C_A C_F-8 C_F T_F n_{\light}+\bigg(\frac{96 \mzv{3}}{5}-\frac{4189}{810}\bigg) C_F{}^2\bigg]
\end{align}
\begin{subequations}
\begin{align}
\num[\ns]{\psq}{2}{01} &= (\lambda -1)^2 \left(8 \lambda ^4+103 \lambda ^3-5262 \lambda ^2+103 \lambda +8\right)\phantom{\big(} \LineNumber
\num[\ns]{\psq}{2}{02} &= 20 \lambda ^4-80 \lambda ^3+141 \lambda ^2-80 \lambda +20\phantom{\big(} \LineNumber
\num[\ns]{\psq}{2}{03} &= -48 \lambda ^6-450 \lambda ^5-508273 \lambda ^4+1114802 \lambda ^3-508273 \lambda ^2-450 \lambda -48\phantom{\big(} \LineNumber
\num[\ns]{\psq}{2}{04} &= 8 \lambda ^6+63 \lambda ^5-5697 \lambda ^4+64113 \lambda ^3-76167 \lambda ^2-22050 \lambda +36750\phantom{\big(} 
\end{align}
\end{subequations}
\begin{align}
  C^{\ns,2}_{L,\psq} &=  a_{s,\ren}\frac{4 C_F}{3} + a_{s,\ren}^2 \bigg[C_F T_F n_{\heavy} \bigg(-\frac{32 \hpl{1} \num[ns]{\psq}{L}{01}}{4725 \lambda ^3}+\frac{192 \hpl{0,0} \num[ns]{\psq}{L}{02}}{5 (\lambda -1)^4}-\frac{16 \hpl{0} \num[ns]{\psq}{L}{03}}{4725 (\lambda -1)^3}\LineNoNumber
  &\quad +\frac{8 \num[ns]{\psq}{L}{04}}{14175 (\lambda -1)^2 \lambda ^2}\bigg)+\bigg(\frac{2878}{135}-\frac{32 \mzv{3}}{5}\bigg) C_A C_F-\frac{184}{27} C_F T_F n_{\light}+\bigg(\frac{64 \mzv{3}}{5}-\frac{1906}{135}\bigg) C_F{}^2\bigg]
\end{align}
\begin{subequations}
\begin{align}
\num[\ns]{\psq}{L}{01} &= (\lambda -1)^6\phantom{\big(} \LineNumber
\num[\ns]{\psq}{L}{02} &= \lambda ^2\phantom{\big(} \LineNumber
\num[\ns]{\psq}{L}{03} &= 2 \lambda ^6-18 \lambda ^5+72 \lambda ^4-693 \lambda ^3-7623 \lambda ^2-7875 \lambda -525\phantom{\big(} \LineNumber
\num[\ns]{\psq}{L}{04} &= 12 \lambda ^6-90 \lambda ^5-11783 \lambda ^4-110258 \lambda ^3-11783 \lambda ^2-90 \lambda +12\phantom{\big(}
\end{align}
\end{subequations}
As elementary consistency checks, we verified that the dependence on the gauge
parameter cancels and that all singularities are removed by mass factorisation,
as expected. Moreover, the one--loop results are in perfect agreement with the
Mellin moments obtained from the exact expressions.

Since no exact results are available for all heavy-quark coefficient functions at NLO, the
comparison is extended to the asymptotic region $m^2 \ll Q^2$. In this limit,
the Mellin moments for $C_{2,i}$ are again found to agree with the results
available in the literature, see Ref.~\cite{Bierenbaum:2007qe} where the Mellin space result have been presented in terms of harmonic sums~\cite{Vermaseren:1998uu,Blumlein:1998if}. In addition, they permit a straightforward
determination of the higher--order corrections in $\lambda$.
\begin{align}
C^2_{2,\psq} &= \langle \bar{\light}_{\set{\vector,\vector}} \rangle^\text{q} + a_{s,\ren} \bigg\{ \frac{1}{3}  \langle \bar{\light}_{\set{\vector,\vector}} \rangle^\text{q} C_F\bigg\} + a_{s,\ren}^2 \bigg\{ -\frac{266}{81} \langle \light_{\set{\vector,\vector}} \rangle ^\psq  C_F T_F n_\light \LineNoNumber
 &\qquad + \langle \heavy_{\set{\vector,\vector}} \rangle ^\psq \bigg[ C_F T_F n_\heavy \bigg(-\frac{32}{9} \hpl{0,0}-\frac{80}{27} \hpl{0}+\frac{160}{9} \hpl{0,0} \lambda+\frac{226}{27} \hpl{0} \lambda-\frac{590}{81}\LineNoNumber
 &\qquad -\frac{1664}{45} \hpl{0,0} \lambda^2+\frac{9671}{162} \lambda+\frac{34532}{675} \hpl{0} \lambda^2-\frac{1318871}{10125} \lambda^2\bigg) + \order{\lambda^3} \bigg] + \langle \bar{\light}_{\set{\vector,\vector}} \rangle^\text{q} \bigg[-8 C_F T_F n_\light\LineNoNumber
 &\qquad+C_F T_F n_\heavy \bigg(-\frac{32}{9} \hpl{0,0}-\frac{56}{15} \hpl{0,0} \lambda^2+\frac{92}{27} \lambda-\frac{112}{9} \hpl{0} \lambda-\frac{140}{27} \hpl{0}-\frac{344}{27}-\frac{2482}{225} \hpl{0} \lambda^2\LineNoNumber
 &\qquad+\frac{9689}{2250} \lambda^2\bigg)+C_F C_A \bigg(-\frac{128}{5} \mzv{3}+\frac{3677}{135}\bigg)+C_F^2 \bigg(\frac{96}{5} \mzv{3}-\frac{4189}{810}\bigg) \bigg] + \order{\lambda^3} \bigg\} + \order{a_{s,\ren}^3}
\end{align}
\begin{align}
C^2_{2,\psg} &= a_{s,\ren} \bigg\{ -\langle \light_{\set{\vector,\vector}} \rangle ^\psg T_F n_\light+ \langle \heavy_{\set{\vector,\vector}} \rangle ^\psg \bigg[ T_F n_\heavy \bigg(-1+6 \lambda-12 \lambda^2-\frac{4}{3} \hpl{0}+\frac{20}{3} \hpl{0} \lambda\LineNoNumber
&\qquad -\frac{44}{3} \hpl{0} \lambda^2\bigg) + \order{\lambda^3} \bigg]  \bigg\} + a_{s,\ren}^2 \bigg\{ \langle \light_{\set{\vector,\vector}} \rangle ^\psg \bigg[ C_A T_F n_\light \bigg(-4 \mzv{3}+\frac{115}{162}\bigg)+C_F T_F n_\light \bigg(\frac{32}{5} \mzv{3}\LineNoNumber
&\qquad -\frac{4799}{405}\bigg) \bigg] + \langle \heavy_{\set{\vector,\vector}} \rangle ^\psg \bigg[ C_A T_F n_\heavy \bigg(-4 \mzv{3}-\frac{8}{5} \mzv{3} \lambda-\frac{11}{162}-\frac{16}{5} \mzv{3} \lambda^2+\frac{44}{9} \hpl{0,0}-\frac{70}{27} \hpl{0}\LineNoNumber
&\qquad-\frac{190}{9} \hpl{0,0} \lambda+\frac{925}{54} \hpl{0} \lambda+\frac{2168}{45} \hpl{0,0} \lambda^2+\frac{27113}{3240} \lambda-\frac{79799}{675} \hpl{0} \lambda^2-\frac{321337}{40500} \lambda^2\bigg)\LineNoNumber
&\qquad+C_F T_F n_\heavy \bigg(\frac{16}{5} \mzv{3} \lambda+\frac{32}{5} \mzv{3}+\frac{32}{5} \mzv{3} \lambda^2-\frac{32}{9} \hpl{0,0}-\frac{160}{27} \hpl{0}+\frac{262}{3} \hpl{0,0} \lambda+\frac{1901}{18} \hpl{0} \lambda\LineNoNumber
&\qquad+\frac{1961}{405}-\frac{14407}{1080} \lambda-\frac{14728}{45} \hpl{0,0} \lambda^2-\frac{147959}{20250} \lambda^2-\frac{196786}{675} \hpl{0} \lambda^2\bigg) +\order{\lambda^3} \bigg] \bigg\} +\order{a_{s,\ren}^3}
\end{align}
\begin{align}
C^{2,\ns}_{2,\psq} &= \langle \bar{\light}^a_{\set{\vector,\vector}} \rangle^\text{q} + a_{s,\ren} \frac{1}{3}  \langle \bar{\light}^a_{\set{\vector,\vector}} \rangle^\text{q} C_F + a_{s,\ren}^2 \langle \bar{\light}^a_{\set{\vector,\vector}} \rangle^\text{q} \bigg[-8 C_F T_F n_\light \LineNoNumber
&\qquad +C_F T_F n_\heavy \bigg(-\frac{32}{9} \hpl{0,0}-\frac{56}{15} \hpl{0,0} \lambda^2+\frac{92}{27} \lambda-\frac{112}{9} \hpl{0} \lambda-\frac{140}{27} \hpl{0}-\frac{344}{27}-\frac{2482}{225} \hpl{0} \lambda^2\LineNoNumber
&\qquad+\frac{9689}{2250} \lambda^2\bigg)+C_F C_A \bigg(-\frac{128}{5} \mzv{3}+\frac{3677}{135}\bigg)+C_F^2 \bigg(\frac{96}{5} \mzv{3}-\frac{4189}{810}\bigg) \bigg] + \order{\lambda^3} \bigg\} + \order{a_{s,\ren}^3}
\end{align}
\begin{align}
C^2_{L,\psq} &=  a_{s,\ren} \bigg\{ \frac{4}{3}  \langle \bar{\light}_{\set{\vector,\vector}} \rangle^\text{q} C_F\bigg\} + a_{s,\ren}^2 \bigg\{ -\frac{160}{27} C_F T_F n_\light \langle \light_{\set{\vector,\vector}} \rangle ^\psq  C_F T_F n_\light \LineNoNumber
 &\qquad + \langle \heavy_{\set{\vector,\vector}} \rangle ^\psq \bigg[ C_F T_F n_\heavy \bigg(\frac{64}{3} \hpl{0,0} \lambda-\frac{160}{27}+\frac{220}{9} \hpl{0} \lambda-\frac{832}{15} \hpl{0,0} \lambda^2+\frac{2393}{27} \lambda+\frac{6536}{225} \hpl{0} \lambda^2\LineNoNumber
 &\qquad-\frac{773258}{3375} \lambda^2\bigg)+ \order{\lambda^3} \bigg] + \langle \bar{\light}_{\set{\vector,\vector}} \rangle^\text{q} \bigg[-\frac{184}{27} C_F T_F n_\light+C_F T_F n_\heavy \bigg(-32 \hpl{0} \lambda-\frac{16}{9} \hpl{0}\LineNoNumber
 &\qquad-\frac{184}{27}+\frac{192}{5} \hpl{0,0} \lambda^2-\frac{680}{9} \lambda-\frac{2912}{25} \hpl{0} \lambda^2-\frac{170192}{1125} \lambda^2\bigg)+C_F C_A \bigg(-\frac{32}{5} \mzv{3}+\frac{2878}{135}\bigg)\LineNoNumber
 &\qquad+C_F^2 \bigg(\frac{64}{5} \mzv{3}-\frac{1906}{135}\bigg)\bigg] + \order{\lambda^3} \bigg\} + \order{a_{s,\ren}^3}
\end{align}
\begin{align}
C^2_{L,\psg} &= a_{s,\ren} \bigg\{ \frac{4}{3}\langle \light_{\set{\vector,\vector}} \rangle ^\psg  T_F n_\light + \langle \heavy_{\set{\vector,\vector}} \rangle ^\psg \bigg[ T_F n_\heavy \bigg(8 \hpl{0} \lambda-16 \hpl{0} \lambda^2+\frac{4}{3}+\frac{20}{3} \lambda-\frac{64}{3} \lambda^2\bigg) \LineNoNumber
&\qquad+\order{\lambda^3} \bigg]  \bigg\} + a_{s,\ren}^2 \bigg\{ \langle \light_{\set{\vector,\vector}} \rangle ^\psg \bigg[ \frac{346}{27} C_A T_F n_\light+C_F T_F n_\light \bigg(-\frac{32}{5} \mzv{3}-\frac{232}{135}\bigg)\bigg] \LineNoNumber
&\qquad+ \langle \heavy_{\set{\vector,\vector}} \rangle ^\psg \bigg[C_A T_F n_\heavy \bigg(\frac{48}{5} \mzv{3} \lambda-\frac{52}{3} \hpl{0,0} \lambda+\frac{96}{5} \mzv{3} \lambda^2+\frac{168}{5} \hpl{0,0} \lambda^2+\frac{323}{9} \hpl{0} \lambda+\frac{346}{27}\LineNoNumber
&\qquad-\frac{9134}{75} \hpl{0} \lambda^2-\frac{11353}{540} \lambda+\frac{38087}{6750} \lambda^2\bigg)+C_F T_F n_\heavy \bigg(-\frac{16}{9} \hpl{0}-\frac{32}{5} \mzv{3}-\frac{96}{5} \mzv{3} \lambda\LineNoNumber
&\qquad-\frac{192}{5} \mzv{3} \lambda^2+\frac{220}{3} \hpl{0,0} \lambda-\frac{232}{135}+\frac{505}{9} \hpl{0} \lambda-\frac{5704}{15} \hpl{0,0} \lambda^2+\frac{35221}{540} \lambda-\frac{81278}{225} \hpl{0} \lambda^2\LineNoNumber
&\qquad-\frac{587707}{6750} \lambda^2\bigg) +\order{\lambda^3} \bigg] \bigg\} +\order{a_{s,\ren}^3}
\end{align}
\begin{align}
C^{2,\ns}_{L,\psq} &=  a_{s,\ren} \frac{4}{3}  \langle \bar{\light}^a_{\set{\vector,\vector}} \rangle^\text{q} C_F + a_{s,\ren}^2 \langle \bar{\light}^a_{\set{\vector,\vector}} \rangle^\text{q} \bigg[ -\frac{184}{27} C_F T_F n_\light+C_F T_F n_\heavy \bigg(-32 \hpl{0} \lambda\LineNoNumber
&\qquad-\frac{16}{9} \hpl{0}-\frac{184}{27}+\frac{192}{5} \hpl{0,0} \lambda^2-\frac{680}{9} \lambda-\frac{2912}{25} \hpl{0} \lambda^2-\frac{170192}{1125} \lambda^2\bigg)+C_F C_A \bigg(-\frac{32}{5} \mzv{3}\LineNoNumber
&\qquad+\frac{2878}{135}\bigg)+C_F^2 \bigg(\frac{64}{5} \mzv{3}-\frac{1906}{135}\bigg) \bigg] + \order{\lambda^3} \bigg\} + \order{a_{s,\ren}^3}
\end{align}
Here, one should note that the limit $\lambda \to 0$ agrees with the limit $\kappa \to 0$ only to leading power $\order{\lambda^0}$. 
For higher powers, one must use the expansions
\begin{align}
  \lambda &= \kappa -2 \kappa ^2+5 \kappa ^3+\order{\kappa ^4}
\, , 
 \LineNumber
  \log{\lambda} &= \log (\kappa )-2 \kappa +3 \kappa ^2-\frac{20 \kappa ^3}{3}+\order{\kappa ^4} 
\, .
\end{align}

\subsection{Numerical results for fixed Mellin moments}

A further non‑trivial consistency check comes from a comparison of our results (hereafter KMS) with the semi‑analytical results of Ref.~\cite{Riemersma:1994hv} (RSvN), which provide parametrisations of the heavy‑quark coefficient functions as functions of $z$ and $\xi$. 
Since these parametrisations are given in $z$‑space, comparison with KMS requires a numerical integration to obtain a Mellin moment. 
This integration can easily be performed with sufficient precision (about six‑digit accuracy). Any deviations between KMS and RSvN therefore test the accuracy of the latter parametrisations.
From the numerical data in Tables \ref{tab:C2qPS}, \ref{tab:C2qNS}, \ref{tab:C1g}, and~\ref{tab:C2g}, it is evident that the leading‑order results are in perfect agreement, and the NLO results show very good overall consistency with the semi‑analytical expressions. 
In the gluon channel, where no independent comparison has been available until now, the agreement reaches the permille level, which corresponds to the accuracy of the RSvN parametrisations. 
We have adopted again the scale choice $\mu_r=\mu_f=Q$ for the moments of coefficient functions in all Tables.

To obtain a meaningful comparison of the quark non-singlet contributions $C_{2,{\rm q}}^{N,{\rm ns}}$ with the RSvN results, the heavy-quark contribution to the form factor (see, e.g. Refs.~\cite{Chuvakin:1999nx,Blumlein:2016xcy}) must be subtracted twice from our result. 
Concretely, to NLO accuracy the quantity compared to RSvN is
\begin{align}
\label{eq:F2ns-subtract}
    C_{2,{\rm q}}^{N,{\rm ns}}
    \biggl|_{\rm RSvN} = 
    C_{2,{\rm q}}^{N,{\rm ns}}
    \biggl|_{\rm KMS} - 2 F_{\rm h}^{(2)}(\kappa)
    \qquad\qquad
    {\rm{at}}~\mu_r=\mu_f=Q
    \, ,
\end{align}
where $F_{\rm h}^{(2)}(\kappa)$ (being proportional to $\delta(1-z)$, i.e. a constant in Mellin $N$-space) denotes the two-loop heavy-quark virtual contribution to the light-quark form factor. 
This subtraction is required because, through the use of the optical theorem, the forward amplitude includes vertex corrections involving heavy quarks in loops. 
These contributions are absent in the RSvN calculation, which includes only open heavy quarks in the final state as required, e.g., for  charm- or bottom-quark structure functions $F_2^{\rm c{\bar c}}$ or $F_2^{\rm b{\bar b}}$. 
Further subtractions are not necessary, since the KMS calculation uses $\alpha_s$ with $n_\light$ flavours, where the heavy quark is decoupled; this is the same renormalization scheme as in RSvN. 
For comparison, Ref.~\cite{Blumlein:2016xcy} employs $\alpha_s$ with $n_\light + n_\heavy$ active flavours, i.e. the heavy quark is treated as active. In that scheme, additional contributions from heavy‑quark self‑energy insertions on external gluon lines must be subtracted. With the decoupled‑flavour definition of $\alpha_s$ used here, these contributions are absent at NLO.
The result for $F_{\rm h}^{(2)}(\kappa)$ in Eq.~(\ref{eq:F2ns-subtract}) found in Ref.~\cite{Blumlein:2016xcy}, with $\beta=\sqrt{1-4 \kappa}$, reads
\begin{align}
F_{\rm h}^{(2)}(\kappa) &= 2 a_{s,\ren}^2  C_F T_F n_\heavy
\Biggl\{\frac{3355}{81}-\frac{952}{9} \kappa+\left(32 \kappa^2-\frac{16}{3}\right)\zeta_3
-\left(\frac{440}{9\xi}-\frac{530}{27}\right)\ln(\kappa)
\LineNoNumber
&
+\beta\left(\frac{184}{9}\kappa-\frac{76}{9}\right)\bigg[\text{Li}_2\left(\frac{\beta+1}
{\beta-1}\right)-\text{Li}_2\left(\frac{\beta-1}{\beta+1}\right)\bigg]+\left(\frac{8}{3}-16 \kappa^2 \right)\bigg[\text{Li}_3\left(\frac{\beta-1}
{\beta+1}\right)\LineNoNumber
&+\text{Li}_3\left(\frac{\beta+1}{\beta-1}\right)\bigg]\Biggr\} .
\end{align}

\input{./tables/numerics.tex}

\section{Conclusions} 
\label{sec:conclusions}

The extension of theoretical computations for heavy‑quark DIS structure functions necessitates the development of new methodological approaches. In the present work, we have computed the NLO QCD corrections to all coefficient functions for $F_2$ and $F_L$, obtaining exact expressions for Mellin moments up to $N=22$ while retaining the full heavy‑quark mass dependence. 
Our results agree with the complete numerical computations available in the literature and correctly reproduce all relevant asymptotic limits, such as the regime $Q^2 \gg m^2$. 
A particular strength of the method is that it starts from the heavy‑quark DIS structure functions themselves, enabling the simultaneous extraction of both the heavy‑quark DIS coefficient functions and the heavy‑quark operator matrix elements. The latter emerge naturally from the factorisation in the limit $Q^2$ much larger than $m^2$.

For the gluon‑initiated NLO coefficient function, we have validated an existing numerical parametrization and found agreement typically at the level of a few permille. The computation has been performed for fixed Mellin moments, and with moderate computational resources we have been able to determine moments up to $N=22$. In addition, we have outlined strategies to extend the computation to even higher moments.
Potential future extensions of the formalism and computational strategy include the treatment of charged‑current processes, i.e. the structure functions $xF_3$ and the analysis of the polarised structure function $g_1$.

The framework developed here can be extended naturally to the three‑loop level, enabling the computation of Mellin moments of heavy‑quark DIS coefficient functions at NNLO. At this order, several technical challenges arise. 
The IBP reduction of the full three‑loop integral families is expected to be a major computational bottleneck.
One possible strategy is to perform the expansion around $\omega = 1/z \to 0$ from the outset; however, this introduces additional complications.
The amplitudes become considerably larger due to the expansion, a systematic tensor decomposition is required to remove the dependence on the external momentum $p$ in the numerator, and the ensuing IBP reduction of three‑loop topologies with internal masses, multiple dots, and non-trivial numerators remains highly demanding.
Although such an approach may restrict the analysis to relatively low Mellin moments, it remains a promising avenue for future developments.
The generation of the three‑loop amplitudes must also be optimised by minimising the number of integral families and choosing them such that only a small subset of denominators must the expanded to recover the $N$-dependence. 
This streamlines the expansion procedure and yields a compact and manageable set of integrals.

Elliptic structures are expected to arise in the computation of Mellin moments, as certain contributions reduce to the three‑loop banana topology.
For the neutral current, these elliptic integrals should occur only in the gluon channel, since quark‑initiated diagrams cannot involve a cut with four heavy quarks.

We have thus presented a set of promising new developments for the computation of heavy‑quark DIS observables, and the corresponding analytic results are provided in the ancillary files together with the submission to \url{www.arXiv.org}.

\vspace{2mm}
\section*{Acknowledgments}
We would like to thank Joshua Davies and Giulio Falcioni for useful discussions.
The work M.K. and S.M. has been supported by the European Research Council 
through  ERC Advanced Grant 101095857, {\it Conformal-EIC}.
The research of K.~S.~was
funded by the European Union’s Horizon Research and Innovation Program under the Marie Skłodowska-Curie grant agreement No.~101204018.

\appendix
%
%
\renewcommand{\theequation}{\ref{sec:appA}.\arabic{equation}}
\setcounter{equation}{0}

\section{Flavour Factors}
\label{sec:appA}

\appendix
\label{app:flavourfactors}
To classify the possible flavour structures, we consider Feynman graphs with up to two quark loops and at most one open light-quark line, coupled to two external vector bosons.
These graphs fall into two broad categories: those involving flavour mixing between light and heavy quarks, and those in which quark mass is preserved along each quark line.
We exclude charged-current interactions and restrict our focus to neutral currents, which preserve flavour along the quark lines.
All QCD vertices are taken to be flavour-preserving.
Considering all possible ways of attaching two bosons, each being either a vector or an axial-vector boson, to quark lines that may be heavy or light, we obtain the fifteen flavour classes.

We formalise the above analysis by examining the flavour structures permitted by the global flavour symmetry groups in the case of neutral currents, constructing an explicit representation using linear algebra (see Ref.~\cite{lang1987linear}), in particular via direct sums. 
To this end, we consider $n_\light$ light-quark flavours and $n_h$ heavy-quark flavours, associated with the symmetry groups $\mathrm{SU}(n_\light)$ and $\mathrm{SU}(n_\heavy)$, generated by $\lambda^a_\light$ and $\lambda^a_\heavy
$, respectively.
The light- and heavy-quark flavours are then embedded into a single unified flavour vector, defined in an enlarged flavour space,
\begin{align}
  \flavour &= \big( \light_1,...,\light_{n_\light},h_1,\dots,h_{n_h} \big) \period
\end{align}
In this representation, the flavour group generators appear as direct sums, each extended by zero matrices so as to act within the enlarged flavour space,
\begin{subequations}
\begin{align}
  \hat{\lambda}_\light &= \lambda^a_\light \oplus 0_{n_\heavy \times n_\heavy} \comma \LineNumber
  \hat{\lambda}^{a}_\heavy &= 0_{n_\light \times n_\light} \oplus \lambda^a_\heavy \period
\end{align}
\end{subequations}
Here, we adopt the convention of indicating embeddings into the extended flavour space with a circumflex.
The charge matrices $Q_{i,\light}$ and $Q_{i,\heavy}$ describe the couplings of vector bosons to light and heavy quarks; these are diagonal for neutral currents, as they preserve quark flavour.
For neutral currents mediated by photons or $Z$-bosons, charge matrices exist for both vector and axial-vector couplings, and we decompose them as follows
\begin{subequations}
\begin{align}
  \hat{Q}_{i,\light} &= Q_{i,\light} \oplus 0_{n_\heavy \times n_\heavy} \comma \quad i \in \set{\axial, \vector} \comma \LineNumber
  \hat{Q}_{i,\heavy} &= 0_{n_\light \times n_\light} \oplus Q_{i,\heavy} \comma \quad i \in \set{\axial, \vector} \period
\end{align}
\end{subequations}
Consequently, the total charge matrices in the enlarged flavour space are obtained by adding the light– and heavy–flavour charge matrices,
\begin{align}
  Q_{i} &= \hat{Q}_{i,\light} + \hat{Q}_{i,\heavy} \comma \quad i \in \set{\axial, \vector} \period
\end{align}
For notational convenience, the embeddings of the light– and heavy–flavour unit matrices within the extended flavour space are also introduced,
\begin{subequations}
\begin{align}
  \hat{1}_\light &= 1_{n_\light \times n_\light} \oplus 0_{n_\heavy \times n_\heavy} \LineNumber
  \hat{1}_\heavy &= 0_{n_\light \times n_\light} \oplus 1_{n_\heavy \times n_\heavy} 
\end{align}
\end{subequations}
In direct analogy with the total charge matrix, we define the total unit matrix in the extended flavour space as the sum of the embedded light- and heavy-flavour unit matrices,
\begin{align}
  1_\flavour = \hat{1}_\light + \hat{1}_\heavy \period
\end{align}

With this notation in place, the flavour structures arising in processes involving two vector bosons mediated by neutral currents can be analysed.
To identify all possible structures, the combinatorics of multiplying charge or unit matrices and taking traces over them is systematically examined.
For this purpose, a shorthand notation is introduced for traces over quark lines to which either a vector or an axial–vector boson is coupled.
In this notation, overlined symbols indicate external quark lines, while non-overlined symbols denote internal quark loops,
\begin{subequations}
\begin{align}
f_B &\equiv \tr \left( \prod_{i \in B}  Q_{i,f} \right) \comma && B \subseteq \set {\vector,\axial}\comma && f \in \set {\light,\heavy} \comma \LineNumber
\bar{f}^a_{\bar{B}} &\equiv \frac{1}{T_{R.f}} \tr \left( \prod_{i \in \bar{B}}  Q_{i,f} \lambda_f^a \right) \hat{\lambda}_{f}^a  \comma && \bar{B} \subseteq \set {\vector,\axial}\comma && f \in \set {\light,\heavy} \comma \LineNumber
\bar{f}_{\bar{B}} &\equiv  \frac{1}{n_{f}}  \tr \left( \prod_{i \in \bar{B}}  Q_{i,f} \right) \hat{1}_{f} \comma && \bar{B} \subseteq \set {\vector,\axial}\comma && f \in \set {\light,\heavy} \comma
\end{align}
\end{subequations}
where $T_{R,f}$ denotes the trace normalisation of the flavour $\mathrm{SU}(n_f)$ generators. The subscripts encode the types of bosons attached to a given quark line. These sets are symmetric in all their entries, as the quark charge matrices are diagonal and therefore commute with one another. Consequently, the ordering of bosons along a given quark line is irrelevant for the flavour factors.
A list notation is therefore employed to collect the traces over quark charge matrices,
\begin{subequations}
\begin{align}
\left\langle x_1, \dots,x_n \right\rangle^\psg &\equiv  \prod_{i = 1}^n x_i \comma \LineNumber
\left\langle x_1, \dots, x_n \right\rangle^\psq &\equiv 1_\flavour \prod_{i = 1}^n x_i \period
\end{align}
\end{subequations}
The entries of these lists may thus be arranged in any order, as the product is fully commutative for a single open quark line, which is the situation under consideration.
Furthermore, for products of matrices that are not traced over, we use the fact that any complex-valued $n_f \times n_f$ matrix admits a decomposition in terms of the generators of $\mathrm{SU}(n_f)$.
Such decompositions equally apply to the embedded matrices in the enlarged flavour space,
\begin{align}
\prod_{i \in B} \hat{Q}_{i,f} &= \frac{1}{n_{f}} \tr \left( \prod_{i \in B} \hat{Q}_{i,f}  \right) \hat{1}_f + \frac{1}{T_{R,f}}  \sum_a \tr \left( \prod_{i \in B} \hat{Q}_{i,f}  \hat{\lambda}_{f}^a \right) \hat{\lambda}_{f}^a \LineNoNumber
&= \frac{1}{n_{f}} \tr \left( \prod_{i \in B} Q_{i,f}  \right) \hat{1}_f + \frac{1}{T_{R,f}}  \sum_a \tr \left( \prod_{i \in B} Q_{i,f} \lambda_{f}^a \right) \hat{\lambda}_{f}^a \LineNoNumber
&= \bar{f}_{\bar{B}} \hat{1}_f  + \bar{f}^a_{\bar{B}}  \hat{\lambda}_{f}^a \comma \qquad \bar{B} \subseteq \set {\vector,\axial}\comma \qquad f \in \set {\light,\heavy} \period \phantom{\left( \prod_{i \in B} Q_{i,f}  \right)}
\end{align}

To analyse the possible flavour structures, we begin with the simplest case, flavour factors consisting solely of traces, without any accompanying matrices.
The first example is the trace over a product of two quark charge matrices,
\begin{align}
	\langle \flavour_{\set{i,j}} \rangle^\psg &\equiv \tr (Q_i Q_j ) \LineNoNumber
    &= \langle \light_{\set{i,j}} \rangle^\psg  + \langle \heavy_{\set{i,j}} \rangle^\psg \period
\end{align}
Here, we have used the fact that the charge matrices are diagonal by assumption.
Since diagonal matrices commute, this flavour factor is symmetric in $i$ and $j$.
It corresponds to the graphs (1-2) in \autoref{fig:flavour_graphs} in which both vector bosons are attached to the same quark loop.
On the other hand, we have the product of traces over individual quark charge matrices,
\begin{align}
  \langle \flavour_{\set{i}} \flavour_{\set{j}} \rangle ^\psg &\equiv \tr (Q_i)\tr (Q_j) \LineNoNumber
  &= \langle \light_{\set{i}} \light_{\set{j}} \rangle ^\psg + \langle \light_{\set{i}} \heavy_{\set{j}} \rangle ^\psg + \langle \heavy_{\set{i}} \light_{\set{j}} \rangle ^\psg +\langle \heavy_{\set{i}} \heavy_{\set{j}} \rangle ^\psg \period
\end{align}
The flavour factors are symmetric in $i$ and $j$ only if the corresponding quark lines are of the same type, since the charge matrices may differ between light and heavy quark flavours.
These structures correspond to graphs (3-6)  in \autoref{fig:flavour_graphs}, where the vector bosons are attached to separate quark loops.
Both former and latter flavour factors arise only in processes involving external gluons or ghosts.

Similar flavour structures emerge for external quarks, now multiplied with the unit matrix to account for quark flavours.
On the one hand, we have the trace of a product of quark charge matrices multiplied with the unit matrix,
\begin{align}
\langle \flavour_{\set{i,j}} \rangle ^\psq &\equiv \tr (Q_i Q_j ) \, 1_\flavour \LineNoNumber
  &= \langle \light_{\set{i,j}} \rangle ^\psq + \langle \heavy_{\set{i,j}} \rangle ^\psq \period
\end{align}
These flavour factors correspond to graphs (7–8) in \autoref{fig:flavour_graphs}, which feature an open quark line with both vector bosons attached to the same quark loop.
On the other hand, we have the product of traces of individual quark charge matrices, each multiplied by the unit matrix,
\begin{align}
\langle \flavour_{\set{i}} \flavour_{\set{j}}  \rangle ^\psq &\equiv \tr (Q_i)\tr (Q_j) \, 1_\flavour\LineNoNumber
  &= \langle \light_{\set{i}} \light_{\set{j}}  \rangle ^\psq + \langle \light_{\set{i}} \heavy_{\set{j}}  \rangle ^\psq + \langle \light_{\set{i}} \heavy_{\set{j}}  \rangle ^\psq + \langle \heavy_{\set{i}} \heavy_{\set{j}}  \rangle ^\psq \period
\end{align}
These flavour factors correspond to graphs (9–12) in \autoref{fig:flavour_graphs}, which feature an open quark line with the two vector bosons attached to separate quark loops.

Additional flavour structures can be constructed by not tracing over all quark charge matrices.
In the case of a product of two charge matrices without a trace, we find that
\begin{align}
  \langle \bar{\flavour}_{\set{i,j}} \rangle^\text{q} &\equiv Q_i Q_j \LineNoNumber
  &=\langle \bar{\light}^a_{\set{i,j}} \rangle^\text{q} \hat{\lambda}^a_\light + \langle \bar{\light}_{\set{i,j}} \rangle^\text{q} \hat{1}_\light + \langle \bar{\heavy}^a_{\set{i,j}} \rangle^\text{q} \hat{\lambda}^a_\heavy + \langle \bar{\heavy}_{\set{i,j}} \rangle^\text{q} \hat{1}_\heavy 
\period
\end{align}
Here, we have employed the decomposition of the quark charge matrices.
This flavour factor is symmetric in $i$ and $j$, owing to the commutativity of diagonal matrices.
It corresponds to graph (13) in \autoref{fig:flavour_graphs}, where both vector bosons are attached to the same open quark line.
The final flavour structure arises from taking the trace over only a single charge matrix,
\begin{align}
  \langle \bar{\flavour}_{\set{i}}
  \flavour_{\set{j}} \rangle ^{\text{q}} 
  &\equiv \hat{Q}_i  \tr \big( \hat{Q}_j \big) \LineNoNumber
  &=  \langle \bar{\light}^a_{\set{i}}
  \light_{\set{j}} \rangle ^{\text{q}} \hat{\lambda}^a_\light + \langle \bar{\light}_{\set{i}}
  \light_{\set{j}} \rangle ^{\text{q}} \hat{1}_\light + \langle \bar{\light}^a_{\set{i}}
  \heavy_{\set{j}} \rangle ^{\text{q}} \hat{\lambda}^a_\light + \langle \bar{\light}_{\set{i}}
  \heavy_{\set{j}} \rangle ^{\text{q}} \hat{1}_\light + \langle \bar{\heavy}^a_{\set{i}}
  \heavy_{\set{j}} \rangle ^{\text{q}} \hat{\lambda}^a_\heavy \LineNoNumber
  &\qquad + \langle \bar{\heavy}_{\set{i}}
  \heavy_{\set{j}} \rangle ^{\text{q}} \hat{1}_\heavy + \langle \bar{\heavy}^a_{\set{i}}
  \light_{\set{j}} \rangle ^{\text{q}} \hat{\lambda}^a_\heavy + \langle \bar{\heavy}_{\set{i}}
  \light_{\set{j}} \rangle ^{\text{q}} \hat{1}_\heavy 
\, ,
\end{align}
where the matrices have been decomposed as for the last flavour structure.
By construction, this flavour factor is not symmetric.
It corresponds to the graphs (14–15)  in \autoref{fig:flavour_graphs}, where the vector bosons are attached to both a quark loop and an open quark line, respectively.

All flavour structures involving external quark lines contain contributions from both light- and heavy-quark flavours.
Since we restrict our discussion to external light quarks, we project onto the corresponding subspace by multiplying with the embedded unit matrix of the light-flavour space.
This removes all external heavy-quark contributions, leaving only light-quark flavours as external lines.

The flavour factors obtained from this analysis coincide precisely with those derived from the combinatorial analysis of trace structures and from the graph-theoretic classification, providing a consistency check that confirms all relevant flavour structures have been systematically accounted for.

Our next step is to analyse flavour factors involving traces over flavour-group generators.
Quarks, whether light or heavy, are partially characterised by their quantum numbers, electric charges and weak isospins.
These quantum numbers naturally partition the flavour space into up-type and down-type sectors, inducing a decomposition of any given set of flavours $\mathcal{F}$ into subsets $\mathcal{U}$ and $\mathcal{D}$.
The trace of a matrix in flavour space is defined by
\begin{align}
\tr_\mathcal{F} ( \lambda ) = \sum_{f\in \mathcal{F}} \lambda_{ff} \comma
\end{align}
where $\lambda$ is an arbitrary flavour matrix.
This definition is equivalent to the standard trace over flavour matrices but allows the summation to be trivially restricted to a subset of flavours.
Using this notation, we prove the following lemma, which enables us to express traces over all flavours involving products of charge matrices and flavour generators in terms of traces over up-type flavours only.
\begin{lemma*}
Let $\mathcal{F}$ be a the disjoint union of the finite sets $\mathcal{U}$ and $\mathcal{D}$, and $\lambda$ be a traceless matrix satisfying $\tr_\mathcal{F}(\lambda) = 0$. Then, 
\begin{align}
    \sum_{\flavour \in \mathcal{F}} c_\flavour \lambda_{\flavour  \flavour} = (c_\mathcal{U}-c_\mathcal{D}) \tr_\mathcal{U} (\lambda) \quad \text{if} \quad c_\flavour  = \begin{cases}
    c_\mathcal{D}, & \text{if } \flavour  \in \mathcal{U} \comma\\
    c_\mathcal{U}, & \text{if } \flavour  \in \mathcal{D} \period
    \end{cases}
\end{align}
\end{lemma*}
\begin{proof} To verify this identity, one proceeds by direct calculation, using the fact that the index set of a finite sum can be partitioned into disjoint subsets. This way, one finds that
\begin{align}
    \sum_{\flavour \in \mathcal{F}} c_\flavour \lambda_{\flavour \flavour} &= \sum_{\flavour \in \mathcal{U}} c_\mathcal{U} \lambda_{\flavour \flavour} + \sum_{\flavour \in \mathcal{D}} c_\mathcal{D} \lambda_{\flavour  \flavour} \LineNoNumber
    &= \sum_{\flavour \in \mathcal{U}} \big( (c_\mathcal{U} - c_\mathcal{D}) + c_\mathcal{D} \big) \lambda_{\flavour  \flavour} + \sum_{\flavour \in \mathcal{D}} c_\mathcal{D} \lambda_{\flavour \flavour} \LineNoNumber
    &= (c_\mathcal{U} - c_\mathcal{D}) \sum_{f\in \mathcal{U}} \lambda_{\flavour \flavour} + c_\mathcal{D} \sum_{\flavour \in \mathcal{D}\sqcup \mathcal{U}}\lambda_{\flavour \flavour} \LineNoNumber
    &= (c_\mathcal{U} - c_\mathcal{D}) \sum_{\flavour \in \mathcal{U}} \lambda_{\flavour \flavour} + c_\mathcal{D} \sum_{\flavour \in \mathcal{F}}\lambda_{\flavour \flavour} \LineNoNumber
    &= (c_\mathcal{U} - c_\mathcal{D}) \tr_\mathcal{U} (\lambda) + c_\mathcal{D} \tr_\mathcal{F} (\lambda) \phantom{ \sum_{\flavour \in \mathcal{F}}} \LineNoNumber
    &= (c_\mathcal{U} - c_\mathcal{D}) \tr_\mathcal{U} (\lambda)\period  \phantom{ \sum_{\flavour \in \mathcal{F}}} 
\end{align}
In the penultimate line, we use that the matrix is assumed to be traceless.
An analogous identity, differing only by an overall sign, follows for the complementary subset upon interchanging the two subsets.
\end{proof}

The foregoing analysis of the flavour factors reveals that the hadronic invariants naturally decompose into non-singlet and pure-singlet contributions.
For external quark states, the corresponding flavour structure takes the form
\begin{align}
    T_{i} (p,q) = \sum_{a} T^a_i  (p,q) \hat{\lambda}^{a}_l + T^{\ps}_i (p,q) \hat{1}_l \period
\end{align}
External gluons require no such decomposition, as these receive only pure-singlet contributions.
To extract the non-singlet component for external quarks, we multiply by a light-flavour group generator $\lambda^a_\light$ and then take the trace of the product.
In contrast, taking the trace without the insertion of a generator directly yields the pure-singlet contribution.
Accordingly, we define the non-singlet and pure-singlet components of the hadronic invariants as
\begin{subequations}
\begin{align}
    T_{i}^{\ns} (p,q) &\equiv \sum_a \frac{1}{T_{F,\light}} \tr \big(T_i  (p,q) \hat{\lambda}^{a}_\light \big) \LineNoNumber
    &= \sum_{N} \left( \frac{1}{x} \right)^N\sum_{a} \frac{\ome{N}{a} (p) \coef{N}{i}[a] (q) }{T_{F,\light}}\comma
    \LineNumber
    T^{\text{ps}}_i  (p,q) &\equiv \frac{1}{n_{\light}}\tr \big(T_i  (p,q) \big) \LineNoNumber
    &= \sum_{N} \left( \frac{1}{x} \right)^N\sum_{k \in \set{\psq,\psg}} \frac{ \ome{N}{a} (p) \coef{N}{n_{\light}}[k] (q)}{n_{\light}} \period
\end{align}
\end{subequations}
In the non-singlet case, the coefficient functions are linear monomials in the flavour factors.
Denoting by $\mathcal{F}_{a}$ the set of all non-singlet flavour factors, this observation implies
\begin{align}
    T^{\ns}_{i} (p,q) &=  \sum_{N} \left( \frac{1}{x} \right)^N  \sum_{a} \frac{\ome{N}{a} (p) \coef{N}{i}[a] (q) }{T_{F,\light}} \LineNoNumber
    &=  \sum_N \left( \frac{1}{x} \right)^N \sum_a \sum_{F_a \in \mathcal{F}_a} \frac{A^N_a (p) F_a }{T_{F,\light}} \dv{C^N_{i,a} (q)}{F_a} \LineNoNumber
    &=  \sum_N \left( \frac{1}{x} \right)^N \left( \sum_a \frac{\tr_{\mathcal{U}} (\hat{\lambda}^a_\light) }{T_{F,\light}} A^N_a (p) \right) \left( \sum_{F_a \in \mathcal{F}_a}  \frac{F_a}{\tr_{\mathcal{U}} (\hat{\lambda}^a_\light)} \dv{C^N_{i,a} (q)}{F_a} \right) \LineNoNumber
    &\equiv  \sum_N\left( \frac{1}{x} \right)^N \ome{N}{\ns} (p) \coef{N}{i}[\ns] (q) \period
\end{align}
Here, use has been made of the fact that both the derivatives of the coefficient functions and the ratios of flavour factors are independent of the flavour generators.
It follows that the non-singlet coefficient functions $C^N_{i,\mathrm{ns}}$ do not depend on the flavour index $a$, since all traces over the light-flavour group generators cancel.

This factorisation applies to all flavour structures, and in principle we may always extract a single flavour factor and absorb it into the operator matrix element.
A common convention is to do so when the flavour factor involves a trace over two charge matrices.
This procedure, however, is restricted to the vector-vector and vector-axial structures; it does not extend to the axial-axial case, where the corresponding non-singlet flavour factor vanishes in physical processes.
In the pure-singlet sector, scaling out flavour factors would introduce an artificial asymmetry between light and heavy internal flavours, and is therefore not adopted here.
Conversion to standard conventions remains straightforward, since the differences amount only to overall multiplicative factors.

\bibliographystyle{jhep}
\bibliography{bibliography.bib}

\end{document}

%% file: commands.tex
\newcommand{\LineNoNumber}{\\ \nonumber}
\newcommand{\LineNumber}{\\}

\renewcommand{\vector}{\text{v}}
\newcommand{\axial}{\text{a}}
\newcommand{\light}{\text{l}}
\newcommand{\heavy}{\text{h}}
\newcommand{\flavour}{\text{f}}
\newcommand{\parton}{\text{p}}

\newcommand{\ren}{\text{ren}}
\newcommand{\const}{\text{const}}

\newcommand{\ps}{\text{ps}}
\newcommand{\ns}{\text{ns}}

\newcommand{\comma}{\, ,}
\newcommand{\period}{\, .}

\RenewDocumentCommand{\sp}{O{}mm}{
  \ifthenelse{\equal{#2}{#3}}
    {\def\spresult{#2^2}}
    {\def\spresult{#2 \cdot #3}}
  \ifthenelse{\equal{#1}{}}
    {\ensuremath{\spresult}}
    {\ensuremath{\left( \spresult \right)^{#1}}}
}

\newcommand{\bjorken}{x}

\newcommand{\intn}[2]{\int \dd[#2]{#1}}
\NewDocumentCommand{\integral}{m m m}{\int_{#2}^{#3} \dd{#1}}

\NewDocumentCommand{\vectors}{o m m m m}{%
  \IfNoValueTF{#1}
    {\tensor{#2}{#3{\lambda_{#4}}} \dots \tensor{#2}{#3{\lambda_{#5}}}}%
    {\tensor{#2}{#3{\{ \lambda_{#4}}} \dots \tensor{#2}{#3{\lambda_{#5}\} }} }%
}

\newcommand{\psq}{\text{q}}
\newcommand{\psc}{\text{c}}
\newcommand{\psg}{\text{g}}

\NewDocumentCommand{\coef}{m m O{}}{%
  C^{#1}_{#2\IfNoValueTF{#3}{}{,#3}}%
}

\newcommand{\mzv}[1]{\zeta_{#1}}

\newcommand{\num}[4][]
  {\text{P}_{#2,#3,#4}^{#1}}

\newcommand{\hpl}[1]{\text{H}_{#1}}

\renewcommand{\order}[1]{\mathcal{O}\left(#1\right)}

\NewDocumentCommand{\ome}{m m O{}}{%
  A^{#1}_{#2\IfNoValueTF{#3}{}{,#3}}%
}

\newcommand{\gammaf}[1]{\Gamma\left(#1\right)}

\newcommand{\adim}[2]{\gamma_{\text{#1}}^{(#2)}}

\newcommand{\bfun}[1]{\beta_{#1}}

\newcommand{\cmom}[4]{c^{N,(#3,#4)}_{#1,#2}}
\renewcommand{\bfun}[1]{\beta_{#1}}
\renewcommand{\adim}[3]{\gamma^{N,(#3)}_{#1#2}}
\newtheorem*{lemma*}{Lemma}

%% file: tables/numerics.tex
\renewcommand{\arraystretch}{1.1} 
\setlength{\LTleft}{0pt}
\setlength{\LTright}{0pt}
\begin{longtable}{@{\extracolsep{\fill}}|c|c|c|c|c|@{}}
\hline 
\multicolumn{1}{|c|}{} &
\multicolumn{2}{c|}{$c_{L,q}^{N,(2,0)}$} &
\multicolumn{2}{c|}{$c_{2,q}^{N,(2,0)}$} \\
\hline
$N$ & KMS & RSvN & KMS & RSvN \\
\hline
\hline
\endfirsthead

\hline
\multicolumn{1}{|c|}{} &
\multicolumn{2}{c|}{$c_{L,q}^{N,(2,0)}$} &
\multicolumn{2}{c|}{$c_{2,q}^{N,(2,0)}$} \\
\hline
$N$ & KMS & RSvN & KMS & RSvN \\
\hline
\hline
\endhead

2 & $-2.663968441e-01$ & $-2.6599e-01$ & $-1.819031902e+00$ & $-1.8187e+00$ \\
4 & $-2.291865727e-02$ & $-2.2898e-02$ & $-1.440499763e-01$ & $-1.4401e-01$ \\
6 & $-3.037959700e-03$ & $-3.0392e-03$ & $-2.040592390e-02$ & $-2.0418e-02$ \\
8 & $-5.184054510e-04$ & $-5.1947e-04$ & $-3.811432361e-03$ & $-3.8175e-03$ \\
10 & $-1.027971471e-04$ & $-1.0318e-04$ & $-8.290910647e-04$ & $-8.3117e-04$ \\
12 & $-2.252325062e-05$ & $-2.2641e-05$ & $-1.985751721e-04$ & $-1.9923e-04$ \\
14 & $-5.299818794e-06$ & $-5.3350e-06$ & $-5.082020292e-05$ & $-5.1024e-05$ \\
16 & $-1.315786777e-06$ & $-1.3262e-06$ & $-1.365120445e-05$ & $-1.3714e-05$ \\
18 & $-3.406331187e-07$ & $-3.4374e-07$ & $-3.804753682e-06$ & $-3.8243e-06$ \\
20 & $-9.119679388e-08$ & $-9.2130e-08$ & $-1.091674299e-06$ & $-1.0978e-06$ \\
22 & $-2.509868611e-08$ & $-2.5381e-08$ & $-3.206582357e-07$ & $-3.2258e-07$ \\
\hline
\caption{Comparison of the numerical values of the singlet coefficient function moments at NLO, 
cf.~Eq.(\ref{eq:def-coeffs}), 
for $m = 2 \sqrt{2}\, \text{GeV}$ and $Q = 7\, \text{GeV}$ such that $\kappa = 8 /49$, for the scale choice $\mu_r=\mu_f=Q$.
}
\label{tab:C2qPS} \\
\end{longtable}

\setlength{\LTleft}{0pt}
\setlength{\LTright}{0pt}
\begin{longtable}{@{\extracolsep{\fill}}|c|c|c|c|c|@{}}
\hline
\multicolumn{1}{|c|}{} &
\multicolumn{2}{c|}{$c_{L,q}^{N,\text{NS},(2,0)}$} &
\multicolumn{2}{c|}{$c_{2,q}^{N,\text{NS},(2,0)}$} \\
\hline
$N$ & KMS & RSvN & KMS & RSvN \\
\hline
\hline
\endfirsthead

\hline
\multicolumn{1}{|c|}{} &
\multicolumn{2}{c|}{$c_{L,q}^{N,\text{NS},(2,0)}$} &
\multicolumn{2}{c|}{$c_{2,q}^{N,\text{NS},(2,0)}$} \\
\hline
$N$ & KMS & RSvN & KMS & RSvN \\
\hline
\hline
\endhead

2 & $1.269706065e-01$ & $1.2691e-01$ & $3.828697684e-01$ & $3.8268e-01$ \\
4 & $7.773989042e-03$ & $7.7697e-03$ & $1.899084702e-02$ & $1.8981e-02$ \\
6 & $8.451385732e-04$ & $8.4611e-04$ & $1.953555304e-03$ & $1.9556e-03$ \\
8 & $1.218724089e-04$ & $1.2232e-04$ & $2.757969313e-04$ & $2.7672e-04$ \\
10 & $2.086984108e-05$ & $2.1002e-05$ & $4.678690020e-05$ & $4.7059e-05$ \\
12 & $4.015640039e-06$ & $4.0518e-06$ & $8.963806878e-06$ & $9.0378e-06$ \\
14 & $8.409051715e-07$ & $8.5059e-07$ & $1.873670645e-06$ & $1.8935e-06$ \\
16 & $1.878096423e-07$ & $1.9041e-07$ & $4.182500377e-07$ & $4.2360e-07$ \\
18 & $4.413049803e-08$ & $4.4839e-08$ & $9.829615433e-08$ & $9.9755e-08$ \\
20 & $1.080458841e-08$ & $1.1000e-08$ & $2.408000184e-08$ & $2.4483e-08$ \\
22 & $2.736748775e-09$ & $2.7912e-09$ & $6.104188218e-09$ & $6.2172e-09$ \\
\hline
\caption{
Same as \autoref{tab:C2qPS} for the non-singlet coefficient function moments at NLO. The values for $c_{2,q}^{N,\text{NS},(2,0)}$ in column 3 (KMS) 
are subject to the subtraction in Eq.~(\ref{eq:F2ns-subtract}).
}
\label{tab:C2qNS}\\
\end{longtable}

\setlength{\LTleft}{0pt}
\setlength{\LTright}{0pt}
\begin{longtable}{@{\extracolsep{\fill}}|c|c|c|c|c|@{}}
\hline
\multicolumn{1}{|c|}{} &
\multicolumn{2}{c|}{$c_{L,g}^{N,(1,0)}$} &
\multicolumn{2}{c|}{$c_{2,g}^{N,(1,0)}$} \\
\hline
$N$ & KMS & RSvN & KMS & RSvN \\
\hline
\hline
\endfirsthead

\hline
\multicolumn{1}{|c|}{} &
\multicolumn{2}{c|}{$c_{L,g}^{N,(1,0)}$} &
\multicolumn{2}{c|}{$c_{2,g}^{N,(1,0)}$} \\
\hline
$N$ & KMS & RSvN & KMS & RSvN \\
\hline
\hline
\endhead

2 & $8.867283709e-02$ & $8.8673e-02$ & $5.028907767e-01$ & $5.0289e-01$ \\
4 & $1.103046457e-02$ & $1.1030e-02$ & $6.175564320e-02$ & $6.1756e-02$ \\
6 & $1.917444241e-03$ & $1.9174e-03$ & $1.183561754e-02$ & $1.1836e-02$ \\
8 & $3.966816131e-04$ & $3.9668e-04$ & $2.731109006e-03$ & $2.7311e-03$ \\
10 & $9.137167103e-05$ & $9.1372e-05$ & $6.987846232e-04$ & $6.9878e-04$ \\
12 & $2.264262570e-05$ & $2.2643e-05$ & $1.909719930e-04$ & $1.9097e-04$ \\
14 & $5.916619363e-06$ & $5.9166e-06$ & $5.463077020e-05$ & $5.4631e-05$ \\
16 & $1.609557695e-06$ & $1.6096e-06$ & $1.616052219e-05$ & $1.6161e-05$ \\
18 & $4.519332213e-07$ & $4.5193e-07$ & $4.904449068e-06$ & $4.9044e-06$ \\
20 & $1.301724905e-07$ & $1.3017e-07$ & $1.518757609e-06$ & $1.5188e-06$ \\
22 & $3.829032869e-08$ & $3.8290e-08$ & $4.780440303e-07$ & $4.7804e-07$ \\
\hline
\caption{
Same as \autoref{tab:C2qPS} for the gluon coefficient function moments at leading order.
}
\label{tab:C1g}\\
\end{longtable}

\setlength{\LTleft}{0pt}
\setlength{\LTright}{0pt}
\begin{longtable}{@{\extracolsep{\fill}}|c|c|c|c|c|@{}}
\hline
\multicolumn{1}{|c|}{} &
\multicolumn{2}{c|}{$c_{L,g}^{N,(2,0)}$} &
\multicolumn{2}{c|}{$c_{2,g}^{N,(2,0)}$} \\
\hline
$N$ & KMS & RSvN & KMS & RSvN \\
\hline
\hline
\endfirsthead

\hline
\multicolumn{1}{|c|}{} &
\multicolumn{2}{c|}{$c_{L,g}^{N,(2,0)}$} &
\multicolumn{2}{c|}{$c_{2,g}^{N,(2,0)}$} \\
\hline
$N$ & KMS & RSvN & KMS & RSvN \\
\hline
\hline
\endhead

2 & $3.378659570e+00$ & $3.4037e+00$ & $1.606527957e+01$ & $1.6121e+01$ \\
4 & $7.289865120e-01$ & $7.3158e-01$ & $3.680483878e+00$ & $3.6872e+00$ \\
6 & $1.676777828e-01$ & $1.6809e-01$ & $9.457069517e-01$ & $9.4677e-01$ \\
8 & $4.148549496e-02$ & $4.1576e-02$ & $2.623839998e-01$ & $2.6259e-01$ \\
10 & $1.087276123e-02$ & $1.0897e-02$ & $7.663126763e-02$ & $7.6681e-02$ \\
12 & $2.977571974e-03$ & $2.9849e-03$ & $2.319720217e-02$ & $2.3211e-02$ \\
14 & $8.436717814e-04$ & $8.4599e-04$ & $7.208613114e-03$ & $7.2126e-03$ \\
16 & $2.455922015e-04$ & $2.4633e-04$ & $2.285220881e-03$ & $2.2865e-03$ \\
18 & $7.307379587e-05$ & $7.3315e-05$ & $7.358664768e-04$ & $7.3627e-04$ \\
20 & $2.213903270e-05$ & $2.2218e-05$ & $2.399582635e-04$ & $2.4009e-04$ \\
22 & $6.809936693e-06$ & $6.8358e-06$ & $7.906076590e-05$ & $7.9104e-05$ \\
\hline
\caption{
Same as \autoref{tab:C2qPS} for the gluon coefficient function moments at NLO.
} 
\label{tab:C2g}\\
\end{longtable}